\documentclass{elsarticle}

\usepackage{graphicx}%
\usepackage{multirow}%
\usepackage{amsmath,amssymb,amsfonts}%
\usepackage{amsthm}%
\usepackage{mathrsfs}%
\usepackage[title]{appendix}%
\usepackage{xcolor}%
\usepackage{textcomp}%
\usepackage{manyfoot}%
\usepackage{booktabs}%
\usepackage{algorithm}%
\usepackage{algorithmicx}%
\usepackage{algpseudocode}%
\usepackage{listings}%
\usepackage{ dsfont }
\usepackage{enumitem}            
\usepackage{faktor}

\usepackage{hyperref}
\usepackage{accents}

\newcommand{\mb}[1]{{\mathbf{#1}}}

\newcommand{\utilde}[1]{\underaccent{\tilde}{#1}}

\usepackage[normalem]{ulem}

\usepackage{aliascnt}   

\newtheorem{theorem}{Theorem}[section]
\newaliascnt{proposition}{theorem}
\newaliascnt{corollary}{theorem}
\newtheorem{corollary}[corollary]{Corollary}%
\newaliascnt{lemma}{theorem}
\newtheorem{lemma}[lemma]{Lemma}%

\newtheorem{definition}{Definition}%
\newtheorem{example}{Example}%

\aliascntresetthe{proposition}
\aliascntresetthe{corollary}
\aliascntresetthe{lemma}

\begin{document} 

\title{Geometric flavours of Quantum Field theory on a Cauchy hypersurface. Part I: Gaussian analysis and other Mathematical aspects  }


\author[1,2,3]{Jos\'e Luis Alonso}

\author[1,2,3]{Carlos Bouthelier-Madre}
\author[1,2,3]{Jes\'us Clemente-Gallardo}

\author[1,3]{David Mart\'inez-Crespo} \ead{dmcrespo@unizar.es}

\affiliation[1]{organization={Department of Theoretical Physics, {University of Zaragoza}},
            addressline={{Facultad de Ciencias, Campus San Francisco}}, 
            city={{Zaragoza}},
            postcode={{50009}}, 
            country={{Spain}}}
\affiliation[2]{organization={ {Institute for Biocomputation and Physics of Complex Systems BIFI}},
            addressline={{Edificio I+D-Campus Río Ebro}, {University of Zaragoza}}, 
            city={{Zaragoza}},
            postcode={{50018}}, 
            country={{Spain}}}

\affiliation[3]{organization={{Center for Astroparticles and High Energy Physics CAPA}},
            addressline={{Facultad de Ciencias, Campus San Francisco, {University of Zaragoza}}}, 
            city={{Zaragoza}},
            postcode={{50009}}, 
            country={{Spain}}}


\begin{abstract}
    In this series of papers we aim to provide a mathematically comprehensive framework to the Hamiltonian pictures of quantum field theory in curved spacetimes. Our final goal is to study the kinematics and the dynamics of the theory from the point of differential geometry in infinite dimensions.    
In this first part we introduce the tools of Gaussian analysis in infinite dimensional spaces of distributions. These spaces will serve the basis to understand the Schrödinger and Holomorphic pictures, over arbitrary Cauchy hypersurfaces, using tools of Hida-Malliavin calculus. Here the Wiener-Ito decomposition theorem provides the  QFT particle interpretation. 
Special emphasis is done in the applications to quantization of these tools in the second part of this paper. 
We devote a section to introduce  Hida test functions as a notion of {\it second quantized } test functions. We also analyze of the ingredients of classical field theory modeled as distributions paving the way for quantization procedures that will be analyzed in \cite{alonsoGeometricFlavours2024}.  
\end{abstract}




\maketitle

\tableofcontents
\newpage

\section{Introduction}

The mathematical description of Quantum Field Theory on an arbitrary background has always been a complicated task. During the past century, most efforts focused on the algebraic side \cite{haagAlgebraicApproach1964, streaterPCTSpin1989,doplicherFieldsObservables1969a,doplicherFieldsObservables1969,doplicherLocalObservables1971,doplicherLocalObservables1974,buchholzLocalityStructure1982,buchholzAlgebraicApproach2020,brunettiAdvancesAlgebraic2015}. That approach proved to be very successful in encoding the properties of quantum operators and states in a mathematically consistent framework. Nonetheless, it did not consider other mathematical aspects of the description, as  the differential geometric ones; which were not described in detail either in the other main mathematical approach to QFT, non-commutative geometry (\cite{connesNoncommutativeGeometry1995,connesRenormalizationQuantum2000}), more focused, again, on algebraic tools.   In particular, the geometrical description of the space of fields justifying, for instance, the existence of a Poisson bracket or  a symplectic form, relies on the definition of the suitable differentiable structure on the fields phase-space, which in this case is cumbersome to build. 

 Furthermore, one of our interests in this construction comes from its potential use in a hybrid quantum-classical Hamiltonian description of quantum fields dynamically coupled to the classical gravitational tensor fields \cite{alonsobujHybridGeometrodynamics2024}. The description of gravity largely depends on differential geometric tools. Thus, we tailor our  description of  the quantum scalar field  with those tools too.

 In our work we consider the symplectic characterization of quantum mechanics of \cite{kibbleGeometrizationQuantum1979}. In that work Kibble defines a Poisson structure on the space of pure states of the Schrödinger-Hamiltonian description of  QFT. In that approach one can investigate the origin of different geometric tools  introduced on the classical phase space of fields and quantum phase space of pure states. Thus, we  aim to carefully study the interdependence of each approach to quantization with the complex structure of the classical phase space and, in turn, on the metric tensor defined on the space-time.  The Hamiltonian treatment of this approach justifies why we prefer  a non-covariant description of the quantum fields (although the global model is covariant) and the choice of the 3+1 Hamiltonian formalism for  gravity. We also stick to this representation of gravity and QFT in \cite{alonsobujHybridGeometrodynamics2024}.

Regarding the mathematical foundations, our construction is largely built upon the results of Gaussian analysis and Hida-Malliavin calculus developed during the last two decades of the XXth century to study stochastic variational systems and applied with huge success to model financial systems \cite{hidaBrownianMotion1980,hidaWhiteNoise1993,kuoWhiteNoise1996,obataWhiteNoise1994, huAnalysisGaussian2016,nunnoMalliavinCalculus2009,kondratievGeneralizedFunctionals1996, westerkampRecentResults2003}. These tools allow us to present the Hamiltonian descriptions of pure quantum states in a simplified language. The set of states has been usually studied with different algebraic or analytical tools under the names of  Schrödinger \cite{corichiSchrodingerFock2004, longSchrodingerWave1998,longSchrodingerWave1996, hofmannClassicalQuantum2015,hofmannNonGaussianGroundstate2017,hofmannQuantumComplete2019,eglseerQuantumPopulations2021,agulloUnitarityUltraviolet2015,kozhikkalBogoliubovTransformation2023} and holomorphic pictures 
\cite{oecklSchrodingerRepresentation2012,oecklHolomorphicQuantization2012,oecklAffineHolomorphic2012}.

In this two-part paper, we re-obtain those results but from a geometrical perspective which will allow us to build a description in terms of Poisson structures in the second part of this series  \cite{alonsoGeometricFlavours2024}.
  This geometrical flavor requires of a careful description of the different structures and of the relations between them required to build the final quantum description.  Hence, we want to build a consistent geometric description of  different sets of tools developed in the last years to model QFT. 

  In this first work we will consider the set of necessary mathematical tools to fix notation and  introduce the structures used in the second work of this series. The structure of the current paper is as follows. Firstly, with the objective of providing with the necessary framework to perform integro-differential calculus and differential geometry, in  section \ref{sec:math_preliminaries} we will briefly summarize how to use the structure of nuclear spaces  to build a differentiable structure on the set of fields. Following this discussion, the second mathematical ingredient within section \ref{sec:math_preliminaries} will delve on the necessity of rigged Hilbert spaces to identify these nuclear spaces with suitable subspaces of their duals to be able to map vectors and covectors on those field manifolds. This mechanism provides a practical framework to deal with tensor calculus.  A very important ingredient  that justifies the use of nuclear spaces to model the space will be  integral calculus. In particular the origin, existence and properties of the set of Gaussian measures on the space of fields, which constitute the third mathematical ingredient treated in section \ref{sec:math_preliminaries}. These Gaussian measures are required to build the Hilbert space structure on the set of pure quantum states.  In order to provide the reader with some insight on the nature of this Gaussian-measure based $L^2$ space, we present several of its dense subalgebras of special relevance in QFT. Namely, the algebras of polynomials $\mathcal{P}$, trigonometric exponentials $\mathcal{T}$ and coherent states $\mathcal{C}$.   
 
  After this mathematical contextualization, we will be ready  to relate our framework with the usual approach to QFT through the introduction of several mathematical transforms, borrowed from stochastic calculus and Gaussian analysis. Firstly, the Wiener-Ito decomposition will yield the coefficients in the (bosonic) Fock space of quantum states. Secondly, the Segal-Bargmann transform will relate the  functional holomorphic and Schrödinger  pictures of QFT. Lastly,  we will  see how to use the tools of Malliavin calculus, in particular the  Malliavin derivative and the Skorokhod integral (generalization of Ito integral to non-adapted stochastic processes \cite{nunnoMalliavinCalculus2009}), to play the  role of  annihilation and creation operators in QFT.
 This compendium of mathematical ingredients   is closed in section \ref{sec:Hida}, where we analyze the set of Hida test functions and their role in the rigorous characterization of the pure quantum states in the Schrödinger (and holomorphic) \textit{wave functional picture},  as it is known in the physical literature.

Within the community of stochastic calculus, the relation between the tools of Malliavin calculus and operators in quantum field theory is already noticed in the literature \cite{henry-labordereAnalysisGeometry2009} and  the close relation of white noise analysis and Hida-Malliavin calculus with Feynman integrals has been also studied  in detail \cite{hidaWhiteNoise1993,kuoWhiteNoise1996,westerkampRecentResults2003}.  From the physics point of view, the necessity of Gaussian analysis has been also hinted as a necessary tool to study the Schrödinger picture of QFT \cite{corichiSchrodingerFock2004, longSchrodingerWave1998,longSchrodingerWave1996, hofmannClassicalQuantum2015,hofmannNonGaussianGroundstate2017,hofmannQuantumComplete2019,eglseerQuantumPopulations2021}, the holomorphic picture \cite{oecklHolomorphicQuantization2012,oecklAffineHolomorphic2012} and their relation thorough the integral transforms \cite{oecklSchrodingerRepresentation2012}. Moreover, Gaussian analysis is needed to describe star products in  deformation quantization \cite{ditoStarproductApproach1990} and in algebraic QFT  quasifree states, also called Gaussian states,  are considered the class of states that reproduce the desirable physical properties of a theory \cite{brunettiAdvancesAlgebraic2015} and they can be easily related with the Wick order of the Wiener-Ito decomposition theorem.

Despite the plethora of examples presented above that reflect the importance of Gaussian analysis as a tool for QFT, to the best of our knowledge, the powerful and efficient formalism developed for  stochastic calculus has not been used in full to describe the Hamiltonian pictures of quantum field theory. Moreover, this formalism allows to take this analysis one step further and study, in a systematic way, the dynamics of the theory. The aim of this series of papers is to bring together both approaches to extract a coherent and systematic way  of treating QFT in curved spacetimes from the Hamiltonian perspective.

Once the mathematical tools for the characterization of functional Hilbert spaces are introduced, Section  \ref{sec:CFTQuantizar} analyzes the geometrical structure of the classical phase space of fields. As a result, we conclude that the introduction of a complex structure in that set   determines the Hilbert space of one-particle quantum states. We analyze also the role of the differential geometric properties of the phase space of classical fields with the construction of the quantum theory. This issue becomes crucial when constructing a hybrid quantum-classical model as the one introduced in \cite{alonsobujHybridGeometrodynamics2024}. 

In section \ref{sec:conclusions} we summarize our results and contextualize the usefulness of  these  tools for applications in QFT. Precisely, in the second paper of this series \cite{alonsoGeometricFlavours2024}, we will use these tools to present different methods of quantization and characterize, in geometrical terms, the differentiable structures of the resulting manifold of quantum states. With these, we will obtain a geometrical description of quantum field theory on a curved background as a Hamiltonian system that will encode the definition of the dynamics from first principles.

\section{Mathematical preliminaries}
\label{sec:math_preliminaries}
In the following sections we will be constructing a differential-geometrical description of the phase space of classical field theories  and the quantum phase space of pure states of a quantum field in the sense of \cite{kibbleGeometrizationQuantum1979}. They are, in both cases, infinite dimensional linear spaces. In order to do that, we will require modelling those manifolds on topological vector spaces where we can construct a differential calculus. Besides, in the case of the quantum phase space, we will also require the operators representing our physical magnitudes to be essentially self-adjoint on their domains, which will be submanifolds of the total space.  In some  applications, we may use a Hilbert space as modelling space, but this is not the general case. The most general framework is provided by dense subsets of functions of a Hilbert space, with enough properties to admit a well-behaved notion of continuity and differentiability. The best suited sets of such functions will be Nuclear-Fréchet (NF) spaces, whose main properties are presented below, with a few relevant examples. Furthermore, we will require to be able to define rigorously Borel probability measures on those manifolds, for the definition of quantum states to make sense. Thus, we are forced to consider also strong duals of NF spaces (DNF) as models of geometry. As Gaussian measures solve the Schrödinger equation for free field theories, \cite{ longSchrodingerWave1998, hofmannNonGaussianGroundstate2017, hofmannClassicalQuantum2015}, this will be our preferred choice. 

Furthermore, the construction of tensors on these manifold will require considering both tangent and cotangent spaces. As the dual spaces of the nuclear spaces are spaces of distributions, tangent and cotangent spaces will be  non-isomorphic. In order to define a practical model, we will consider the notion of Gel'fand triple to be able to model vectors and covectors as elements of a suitable Hilbert space. This will provide us with suitable tensors on the tangent or on suitable dense subsets of the cotangent space. 

Finally, we will analyze in detail the nuclear space used to model the set of pure quantum states, the space of Hida test functions, and the corresponding Gel'fand triple. While, in regard to this paper, it is just an illustrative example of the construction for a certain nuclear space, its importance for the analysis in the second part of this series demands a specific analysis. 


\subsection{Nuclear spaces: doing differential geometry in infinite dimensions}
It is a well known fact that, in order to define Gaussian Borel probability measures in infinite dimensions, the natural setting is that of  Nuclear-Fréchet spaces (NF) and its strong duals (DNF). This is discussed thoroughly in \cite{gelfandGeneralizedFunctions1964}.  Besides both categories of topological vector spaces are convenient and admit partitions of unity which turns them into suitable models for infinite dimensional geometry in the sense of \cite{krieglConvenientSetting1997}. See also \cite{dodsonGeometryFrechet2016} for the Frechèt case.  In this section we focus on integration theory and we will summarize its main ingredients on the dual pair  $(\mathcal{N},\mathcal{N}')$ where $\mathcal{N}$ is a NF space and $\mathcal{N}'$ is its dual endowed with the strong topology\footnote{In this case we endow the DNF space with the strong topology, which coincides with the Mackey topology, but the same Borel $\sigma$-algebra is generated from the weak topology. }. In particular, we are interested in
$\big(C^{\infty}_c(K),D'(K)\big)$, the compactly supported functions and distributions over a compact subset $K\subset \mathbb{R}^{d}$ and $\big(\mathcal{S}(\mathbb{R}^d),\mathcal{S}'(\mathbb{R}^d)\big)$ the space of Schwartz functions of fast decay and the space of tempered distributions. 
 
For completeness we define here the functional spaces to be used in a way that fits our purpose.

\begin{definition}
    \label{def:Nuclear}
    {\bf (Nuclear Space)} A nuclear space $X$ is a  topological vector space whose topology is given by a family of Hilbert seminorms $p$  such that, if $X_{p}$ is the Hilbert space obtained as  the completion of the preHilbert  space $(X/p^{-1}(0), p)$, then there is a larger seminorm $q$ such that $X_q\hookrightarrow X_p$ is a Hilbert-Schmidt operator \cite{trevesTopologicalVector1967}. 
\end{definition}

For infinite dimensional topological vector spaces $X,Y$ tensor products  can be endowed with a family of  different topologies between the projective $X\otimes_\pi Y$ and injective $X\otimes_\epsilon Y$ tensor product induced topologies.  If one or both spaces are nuclear both topologies coincide and therefore there is only one notion of tensor product $X\otimes Y$. We will  use this fact in the following sections to define polynomials over $X$.

\begin{definition}
    \label{def:FN}
    {\bf (NF  Nuclear-Fréchet Space)} A Nuclear space which is also Fréchet  $\mathcal{N}$  can be regarded as the projective limit of a countable family of separable\footnote{Every Frechèt space is complete and  Barreled   (\cite{trevesTopologicalVector1967} Corollary 1 Chapter 33). A complete  Barreled Nuclear space is a Montel space   (\cite{trevesTopologicalVector1967} Corollary 3 Chapter 50). Montel spaces are reflexive (\cite{trevesTopologicalVector1967} Corollary of Proposition 36.9) and Frechèt Montel spaces are separable  (\cite{krieglConvenientSetting1997} Result 52.27). Thus it follows that a Nuclear-Frechèt space is reflexive and separable. This fact allows to choose the chain of separable  Hilbert spaces for this definition. } Hilbert spaces $\{\mathcal{H}_p\}_{p=0}^\infty$, with scalar products $\langle \cdot, \cdot  \rangle_p$ \ such that for $q>p$     
    the injection $\mathcal{H}_q\hookrightarrow \mathcal{H}_p$ is Hilbert-Schmidt. Thus
    $$
    \mathcal{N}=\bigcap_{p=0}^\infty\mathcal{H}_p 
    $$
    with the projective limit topology. 
\end{definition}

 Riesz representation theorem allows us to identify a Hilbert space with its dual. By convention we set $\mathcal{H}_0'=\mathcal{H}_0^*$ where $*$ represents complex conjugation. This implies that in the inclusion $\mathcal{H}_p\hookrightarrow\mathcal{H}_0$, for $p>0$, the space $\mathcal{H}_p$ can no longer be identified with its dual. We will denote this dual space by $\mathcal{H}_0\hookrightarrow\mathcal{H}_{-p}$. The next definition can be regarded as a consequence of the last one. 

\begin{definition}
    \label{def:DFN}
    {\bf (DNF Dual of a  Nuclear Fréchet Space)} The strong dual of a NF space  $\mathcal{N}'$  can be regarded as the inductive limit of a countable family of separable Hilbert spaces $\{\mathcal{H}_{-p}\}_{p=0}^\infty$ such that for $q>p$
    the injection $\mathcal{H}_{-q}\hookleftarrow \mathcal{H}_{-p}$ is the transpose of a  Hilbert-Schmidt operator. Thus
    $$
    \mathcal{N}'=\bigcup_{p=0}^\infty\mathcal{H}_{-p} 
    $$
    with the inductive limit topology. A DNF space is always nuclear. Furthermore it is a Nuclear Silva space \cite{krieglConvenientSetting1997}.
\end{definition}

\subsection{Rigged Hilbert spaces and main examples}

With the definitions given so far we notice that the space $\mathcal{H}_0$ plays an important role. This motivates the consideration of a triple of spaces therefore leading to the concept of rigged Hilbert space. 

\begin{definition} {\bf (Rigged Hilbert Space)}
    \label{def:RHS}
    Assume that $\mathcal{H}_0$ is a separable Hilbert space then, then we define a rigged Hilbert space or Gel'fand triple as the triple of spaces
\begin{equation*}
\label{eq:rigged}
\mathcal{N}\subset \mathcal{H}_0 \subset \mathcal{N}'
\end{equation*}
where $\mathcal{N}$ is a NF space and $\mathcal{N}'$ its dual and each inclusion is continuous. 

\end{definition}

In the structure  presented here $\mathcal{H}_0$ is said to realize the duality, which means that the scalar product of $\mathcal{H}_0$ identifies an element of $\mathcal{N}$ with its dual in $\mathcal{N}'$. 
$\mathcal{N}$ is a dense subspace of the Hilbert space and this construction allows to, in some sense, parametrize the duality and define 
  $\mathcal{N}$ as a subset of $\mathcal{N}'$. Notice, however, that in general different Hilbert spaces can be used to parametrize the duality of the same pair $(\mathcal{N},\mathcal{N}')$.

In the rigging there are three main ingredients: the set of functions $\mathcal{N}$, its dual space  $\mathcal{N}'$ and the dual pairing of the Hilbert space $\mathcal{H}_0$. Firstly,  elements  $\xi_{\mathbf{x}}\in\mathcal{N}$ will be typically well defined functions (see examples below) and will be  denoted by a bold subindex.  Secondly,  elements $\varphi^{\mathbf{x}}\in\mathcal{N}'$ will be distributions (with no pointwise meaning in general) and they will be  denoted by a bold superindex. Sometimes we will denote the value of the function at a point $x$  with  a regular subindex $\xi_x$, but notice that this  can not be done in general for distributions. With those two ingredients the duality is denoted, mimicking Einstein summation convention, 
\begin{equation}
\label{eq:dualityNotationDistributions}
\xi_{\mathbf{x}}\in\mathcal{N}, \varphi^{\mathbf{x}}\in\mathcal{N}' \textrm{ paired by } \xi_x\varphi^x :=\langle\xi_{\mathbf{x}}, \varphi^{\mathbf{x}}\rangle \in \mathbb{R}\ \ (\mathbb{C}).
\end{equation}

The third and last ingredient is the dual pairing of $\mathcal{H}_0$. This element can be  considered in a threefold way
\begin{enumerate}[label=\alph*)]
    \item Dual pairing between elements in $\mathcal{N}\subset \mathcal{H}_0$ with $\langle\xi^1_{\mathbf{x}} ,\xi^2_{\mathbf{y}}\rangle =\xi^1_x\Delta^{xy}\xi_y^2$
    \item Dual pairing between elements in  $ \mathcal{H}_0\subset \mathcal{N}'$ with $\langle\varphi_1^{\mathbf{x}} ,\varphi_2^{\mathbf{y}}\rangle =\varphi_1^xK_{xy}\varphi^y_2$ 
    \item Injection of ${\Delta^{\mathbf{xy}}}:\mathcal{N}\hookrightarrow \mathcal{N}'$ trough the pairing of $\mathcal{H}_0$ such that $\xi_x\Delta^{xy}=\xi^y$ with a densely defined inverse $\varphi^xK_{xy}=\varphi_y$ such that $\Delta^{xz}K_{zy}=\delta^x_y$ is the evaluation mapping.
\end{enumerate}
We will be using the three different frameworks in the following sections.

The main examples of NF spaces, which play the roles of model spaces for the classical field theories in this work are the following:

\begin{example}\label{ex:testFunctions}{\bf (Smooth functions of compact support on a compact $K$)}
    Let $K\subset \mathbb{R}^d$ then the space $C^{\infty}_c(K)$ is NF \cite{trevesTopologicalVector1967} with the initial topology induced from the family of Hilbertian norms

\begin{equation}
\label{eq:hilbertianNormsCompactSupport}
\lVert\varphi\rVert_{K,n}^2=\sum_{\lvert\alpha\rvert\leq n}\int_{K}\lvert\partial^\alpha\varphi(x)\rvert^2 dx^d,
\end{equation}
where $\alpha\in \mathbb{N}_0^d$ is a multi-index and $\lvert\alpha\rvert=\alpha_1+\cdots+\alpha_d$.  Its dual is the DNF space $D'(K)$. This is is not the family of Hilbertian seminorms of  \autoref{def:FN}, instead we should use the Hamiltonian of the harmonic oscillator to introduce them as in \cite{hidaWhiteNoise1993}.

Notice that we could also define $C^\infty_c(U)$, for $U\subset \mathbb{R}$ an open set. This is a nuclear  but not  a Fréchet space; instead it is a Nuclear LF space (see \cite{gelfandGeneralizedFunctions1964} for details). For our purposes we will limit our attention to the compact domain in order to reduce  technicalities.
\end{example}

\begin{example}{\bf (Smooth functions of fast decay)}\label{ex:schwartFunctions} The space $\mathcal{S}(\mathbb{R}^d)$ of smooth functions of fast decay is NF  with the initial topology induced from the family of Hilbertian norms

    \begin{equation}
        \label{eq:hilbertianNormsSchwartz}
        \lVert\varphi\rVert_{m,n}^2=\sum_{\lvert\alpha\rvert\leq n, \lvert\beta\rvert\leq m}\int_{\mathbb{R}^d}\prod_{i=1}^d(1+x_i^2)^{2\beta_i}\lvert \partial^\alpha\varphi(x)\rvert^2 dx^d
    \end{equation}
    where $\alpha,\beta\in \mathbb{N}_0^d$ are multi-indices. This is not the family of Hilbertian seminorms of  \autoref{def:FN}, instead we should use the Hamiltonian of the harmonic oscillator to introduce them as in \cite{hidaWhiteNoise1993}. This topology is equivalent to the one usually introduced via supremum type seminorms \cite{hidaBrownianMotion1980}. Its dual $\mathcal{S}'(\mathbb{R}^d)$ is the DNF space of tempered distributions. This space is the one commonly used in white noise analysis \cite{hidaWhiteNoise1993}  
\end{example}

\begin{example}\label{ex:testFunctionsManifold}{\bf (Smooth functions of compact support on a compact Manifold )} For $\Sigma$ a compact manifold the definition of $C^\infty_c(\Sigma)$ goes on as in \autoref{ex:testFunctions}. If we endow the manifold with a Riemannian metric $h$ there is a canonical volume form 
\begin{equation}
\label{eq:riemannianVolumeform}
\int_\Sigma \xi\ \mathrm{dVol}_h=\int_U dx^d \sqrt{\lvert h(x)\rvert} \xi(x),
\end{equation}
for $\xi$ supported on a subset of $\Sigma$ covered by  a single chart with coordinates in $U$. This, in turn, leads to a canonical Hilbert space of square integrable functions $L^2(\Sigma,\mathrm{dVol}_h)$ and therefore to the rigged Hilbert space 
\begin{equation}
\label{eq:riggedOnAManifold}
C^\infty_c(\Sigma)\subset L^2(\Sigma,\mathrm{dVol}_h) \subset D'(\Sigma)
\end{equation}

As discussed in \autoref{def:RHS} the scalar product on a rigged Hilbert space can be considered in a threefold way, as we can illustrate over a chart with coordinates in $U\subset \mathbb{R}^d$. Because of the importance of this particular example throughout this work, we  detail the   notation for the bilocal kernels exposed below \eqref{eq:dualityNotationDistributions}. All three cases will be denoted by the same symbol $\delta^{\mathbf{xy}}$, $\delta_{\mathbf{xy}}$ and $\delta^{\mathbf{x}}_{\mathbf{y}}$ and the difference between them is understood by context as follows 

\begin{enumerate}[label=\alph*)]
        \item Dual pairing between elements in $C^\infty_c(\Sigma)\subset L^2(\Sigma,\mathrm{dVol}_h)$ with $$\langle \xi_{\mathbf{x}} ,\eta_{\mathbf{y}}\rangle =\chi_x\delta^{xy}\eta_y=\int_U dx^d\sqrt{\lvert h(x)\rvert} \xi(x)\eta(x)$$
    \item Dual pairing between elements in  $L^2(\Sigma,\mathrm{dVol}_h) \subset D'(\Sigma)$ with $$\langle\pi^{\mathbf{x}} ,\varphi^{\mathbf{y}}\rangle =\pi^x\delta_{xy}\varphi^y=\int_U\frac{dx^d}{\sqrt{\lvert h(x)\rvert}} \pi(x)\varphi(x),$$
        where $\pi(x),\varphi(x)$ are identified with  densities of weight 1. Notice that the factor $\sqrt{\lvert h(x)\rvert}$ appears in the denominator and $\frac{\pi(x)}{\sqrt{\lvert h(x)\rvert}}, \frac{\varphi(x)}{\sqrt{\lvert h(x)\rvert}}$ are $L^2(\Sigma, \mathrm{dVol}_h)$ representatives.
    \item We can also see the injection ${\delta^{\mathbf{xy}}}:C^\infty_c(\Sigma) \hookrightarrow D'(\Sigma)$  such that
    $$\eta^y=\eta_x\delta^{xy}=\int_U dx^d\sqrt{\lvert h(x)\rvert}\delta^d(x-y) \eta(x)=\sqrt{\lvert h(y)\rvert}\eta(y)$$
    is a representative of a density of weight 1. On the other hand $\delta_{\mathbf{xy}}$ is a densely defined inverse 
    $$\varphi_y=\varphi^x\delta_{xy}=\int_U \frac{dx^d}{\sqrt{\lvert h(x)\rvert}}\delta^d(x-y) \varphi(x)=\frac{\varphi(y)}{\sqrt{\lvert h(y)\rvert}}$$
    such that $\delta^{xz}\delta_{zy}=\delta^x_y$ is the evaluation mapping
    $$\delta^x_y\eta_x=\eta_y=\eta(y)$$
    which only makes sense over $C^\infty_c(\Sigma)$.
\end{enumerate}

The differences between $\delta^{\mathbf{xy}}$, $\delta_{\mathbf{xy}}$ and $\delta^{\mathbf{x}}_{\mathbf{y}}$ are often misinterpreted in the literature since the three of them can be identified with the Dirac delta in $\mathbb{R}^d$ with the Lebesgue measure. 

In this picture elements of  $D'(\Sigma)$ are distributions used to integrate the elements of $C^\infty_c(\Sigma)$ and this nontrivial structure is reframed in the rigged Hilbert space with a nontrivial representation of the scalar product as an integral kernel.

Notice that for different applications, the identification of the field and its dual may be different. Thus, in  classical field theory the common approach \cite{gotayMomentumMaps2004a,gotayMomentumMaps2004} is to model classical fields as elements of a $n$-th jet bundle with $n$ the number of derivatives needed to describe the theory. We can say that, in this approach, fields are modelled over elements of  $C^\infty_c(\Sigma)$. In those approaches the dual jet bundle is built upon the concept of scalar density of weight one  that are constructed using volume forms $\xi dx^1\wedge \cdots\wedge dx^d$ with $\xi\in C^\infty_c(\Sigma)$.  Thus, the topology of the model space for the dual jet bundle is that of  $C^\infty_c(\Sigma)$. This approach makes the treatment significantly easier at the cost of giving up the possibility of describing the duality with the strong topology. In those approaches the dual pairing is described as in a). On the other hand, in  algebraic QFT the situation is radically different. Classical fields to be quantized are taken as a priori distributions of $D'(\Sigma)$   and equations of motion must be implemented weakly using arbitrary test functions acted upon by the adjoint evolution operator.  For this reason we must think of both fields and their duals as  modelled over the same space of distributions   and the dual pairing only makes sense as in case b) \cite{brunettiAdvancesAlgebraic2015}. 

\end{example}

Other important family of examples of rigged Hilbert spaces are reliant upon the NF space of  Hida test functions, which will be the set used to model the space of  pure quantum field states, considered as Schrödinger wavefunctionals with fields in their domain. Its construction is significantly more difficult and needs a careful study of Gaussian integration to be developed in the following sections. For this reason  we postpone its introduction to \autoref{sec:Hida}. 

\subsection{Defining a scalar product: Gaussian measures}
As we are intending to  describe a theory of  pure quantum field states, we must endow the space of quantum states with a scalar product.  As we mentioned above, the most commonly used measures in QFT correspond to Gaussian measures, often referred to a given vacuum state. 
We will review below their main definition and  properties, first for the simple case of real test functions, and then in the case of complex ones, which will be more useful for us later.
\subsubsection{The real case}
The procedure to define a Gaussian measure $\mu$ over a  NF space $\mathcal{N}'$ is straightforward once we notice that the $\sigma$-algebra generated by its cylinder subsets corresponds to the Borel $\sigma$-algebra of the weak and strong topologies of $\mathcal{N}'$  \cite{kuoWhiteNoise1996}. Bochner-Minlos theorem  \cite{hidaBrownianMotion1980,gelfandGeneralizedFunctions1964} provides an even straighter path to defining $\mu$.  First, we consider the bilinear operation ${\xi}_x\Delta^{xy}\xi_y$ to be continuous   with respect to the topology of $\mathcal{N}$.  Then, we consider the  characteristic functional of a Gaussian measure $\mu$
\begin{equation}
    \label{eq:chartacteristicFunctional}
        C(\xi_{\mathbf{x}})=\int_{\mathcal{N}'}D\mu(\varphi^{\mathbf{x}}) e^{i\xi_x\varphi^x} :=e^{-\frac{1}{2}{\xi}_x\Delta^{xy}\xi_y}.
\end{equation}

Bochner-Minlos theorem states that there exists a 
 a unique Borel probability measure that fulfills \eqref{eq:chartacteristicFunctional}. We will use this fact as the definition of $\mu$. In \cite{alonsoGeometricFlavours2024} we will interpret this measure as part of the vacuum.
Considered with respect to a Gel'fand triple,
$$
\mathcal{N}\subset \mathcal{H}_\Delta \subset \mathcal{N}^{\prime },
$$
the bilinear symmetric form $\Delta^{\mathbf{xy}}$ corresponds to the representation of the scalar product of the Cameron-Martin Hilbert space $\mathcal{H}_\Delta$.  This space is obtained by completion of $\mathcal{N}$ with the covariance $\Delta^{\mathbf{xy}}$ and is of huge importance in Gaussian analysis  because provides the necessary tools to develop integration as we will see below. This triple also posses an important physical interpretation in QFT, it represents the Hilbert space of one particle states as will be  explained around  \eqref{eq:HolomorphicTriple}.

\subsubsection{The complex case}
If we consider the case of complex test functions $\mathcal{N}_\mathbb{C}$ of any of the aforementioned examples of nuclear spaces, most of the properties work in the same form as for the real case.  Let us consider the complexification of the space of distributions such that  an element $\phi^{\mathbf{x}}\in\mathcal{N}'_\mathbb{C}$ is written as \begin{equation}
    \label{eq:distHolomorph}
\phi^{\mathbf{x}}=\frac{1}{\sqrt{2}}(\varphi^{\mathbf{x}}-i\pi^{\mathbf{x}}),
\end{equation}  for $(\varphi^{\mathbf{x}},\pi^\mathbf{x})\in \mathcal{N}'\times \mathcal{N}'$ from which follows $\bar{\phi^{\mathbf{x}}}=\frac{1}{\sqrt{2}}(\varphi^{\mathbf{x}}+i\pi^{\mathbf{x}})$. We use this sign convention for the holomorphic coordinate to match with the conventions of \cite{alonsoGeometricFlavours2024}.   Complex test functions can also be written in holomorphic and antiholomorphic  coordinates, for $\rho_{\mathbf{x}}\in\mathcal{N}_\mathbb{C}$ we use the dual conjugate convention \begin{equation}
    \label{eq:testHolomorph}
\rho_\mathbf{x}=\frac{1}{\sqrt{2}}(\xi_\mathbf{x}+i\eta_\mathbf{x})
\end{equation} for $(\xi_{\mathbf{x}},\eta_{\mathbf{x}})\in \mathcal{N}\times\mathcal{N}$ and $\bar\rho_x=\frac{1}{\sqrt{2}}(\xi_\mathbf{x}-i\eta_\mathbf{x})$. With these conventions the dual pairing is written $\langle \rho_{\mathbf{x}},\phi^{\mathbf{x}}\rangle= \rho_x\phi^x$. If we endow each copy of the cartesian product with the measure given by \eqref{eq:chartacteristicFunctional},  we can define a measure to the complex domain  defined by the characteristic functional  

\begin{equation}
\label{eq:complexCharacteristicFunctional}
C_c(\rho_{\mathbf{x}},\bar\rho_{\mathbf{y}})=\int_{\mathcal{N}_{\mathbb{C}}'}D\mu_c(\bar \phi^{\mathbf{x}}, \phi^{\mathbf{x}}) e^{i({\overline{\rho_x{\phi}^x}+\rho_x{\phi}^x})} =e^{-\bar{\rho}_x\Delta^{xy}\rho_y},
\end{equation}
where now the covariance $\Delta^{\mathbf{xy}}$ will be a hermitian form on $\mathcal{N}_\mathbb{C}$, i.e. $\Delta^{\mathbf{xy}}=(\Delta^*)^{\mathbf{xy}}$. 
Notice that, by construction, the complex measure is the product measure of two real ones. In this case the associated Cameron-Martin Gel´fand triple is 
$$
\mathcal{N}_{\mathbb{C}}\subset \mathcal{H}_\Delta \subset \mathcal{N}^{\prime *}_\mathbb{C},
$$
The conjugation $*$ follows from our conventions \eqref{eq:distHolomorph} and \eqref{eq:testHolomorph}. To fix notation we indicate that, throughout  this work, the strong dual will be denoted  by $'$ while $*$ indicates only  complex conjugation.

\subsubsection{Integral transforms on Gaussian measures}
Among the different integral transforms defined over a Gaussian measure the most relevant for Quantum Field Theory is the $S$ transform. Throughout this whole section we will consider a Gaussian measure $\mu$ defined on $\mathcal{N}'$ or $\mathcal{N}'_\mathbb{C}$ with associated characteristic functional \eqref{eq:complexCharacteristicFunctional}. For the real case we define:

\begin{definition}
    \label{def:Stransform}{\bf ($S_\mu$ transform)} For the  Gaussian measure of \eqref{eq:chartacteristicFunctional} we define the $S_\mu$ transform as a generalization of the Laplace transform. For a given function $\Psi:\mathcal{N}'\to \mathbb{C}$ it takes the form 
    \begin{equation}
        \label{eq:stransform}
        S_\mu[\Psi](\xi_\mathbf{x})=\int_{\mathcal{N}'}D\mu(\varphi^{\mathbf{x}})\Psi(\varphi^{\mathbf{x}})\exp{\big(\xi_x\varphi^x-\frac12\xi_x\Delta^{xy}\xi_y\big)}
\end{equation}
\end{definition}
This is to be thought as the action of a translation $\varphi^{\mathbf{x}}\to\varphi^{\mathbf{x}}+\Delta^{\mathbf{x}y}\xi_y$ on the domain of integration. For the complex case it is readily defined in (anti)-holomorphic coordinates as 

\begin{definition}
    \label{def:SComplextransform}{\bf ($S_{\mu_c}$ complex transform)} For the complex Gaussian measure of \eqref{eq:complexCharacteristicFunctional} and a given function $\Psi:\mathcal{N}_{\mathbb{C}}'\to \mathbb{C}$, the transform takes the form 
    \begin{equation}
        \label{eq:stransformcomplex}
        S_{\mu_c}[\Psi](\rho_\mathbf{x},\bar\rho_\mathbf{x})=\int_{\mathcal{N}'_{\mathbb{C}}}D\mu_c(\phi^{\mathbf{x}})\Psi(\phi^{\mathbf{x}},\bar\phi^{\mathbf{x}})\exp{\big(\overline{{\rho}_x\phi^x}+{\rho}_x{\phi}^x-\bar\rho_x\Delta^{xy}\rho_y\big)}
        \end{equation}
\end{definition}

Notice that this definition is nothing but the combination of two copies of \autoref{def:Stransform}, one for each variable, and the required change to holomorphic coordinates.

Another relevant transform is the $T$ transform. It is the  extension of  the definition of characteristic functional \eqref{eq:complexCharacteristicFunctional} to  include a generic function besides the measure, defined as follows
\begin{definition}
    \label{def:TComplextransform}{\bf ($T_{\mu_c}$ complex transform)} For the complex Gaussian measure of \eqref{eq:complexCharacteristicFunctional} and a given $\Psi:\mathcal{N}_{\mathbb{C}}'\to \mathbb{C}$  takes the form 
\begin{equation}
    \label{eq:Ttransform}
    T_{\mu_c}[\Psi](\rho_\mathbf{x},\bar{\rho}_\mathbf{y})=\int_{\mathcal{N}'_\mathbb{C}}D\mu_c(\phi^{\mathbf{x}})e^{i(\overline{\rho_z\phi^z}+\rho_z\phi^z)} \Psi(\phi^{\mathbf{x}},\bar\phi^\mathbf{y})
    \end{equation}
\end{definition}

\noindent Both transforms are related by $C(\rho_\mathbf{x},\bar{\rho}_\mathbf{y}) S_{\mu_c}[\Psi](\rho_\mathbf{x},\bar\rho_\mathbf{x})= T_{\mu_c}[\Psi](-i\rho_\mathbf{x},-i\overline{\rho}_\mathbf{y})$ each of those transformations reveal  different features of $L^2(\mathcal{N}'_\mathbb{C},D\mu_c)$ as we will see bellow.

\subsection{The Hilbert space of Quantum pure states  $L^2(\mathcal{N}'_\mathbb{C},D\mu_c)$ through its subalgebras}
\label{sec:algebras}

The Hilbert space $L^2(\mathcal{N}'_\mathbb{C},D\mu_c)$ will represent in \cite{alonsoGeometricFlavours2024} the Hilbert space of pure quantum states of a {QFT} of the scalar field. In this section we introduce it through the study of some their subalgebras. 
The structure of this  Hilbert space  may seem difficult to tackle down. At first glace we are dealing with functions with domains in an infinite dimensional space of distributions. This is some times referred to as working at the level of {\it second quantization} in the literature. In order to study the whole space, it is useful to consider  families of functions that are dense in the space $L^2(\mathcal{N}'_\mathbb{C},D\mu_c)$ and that simplify the definition of the objects we are interested in. Here we include the algebras of polynomials $\mathcal{P}$, trigonometric exponentials $\mathcal{T}$ and coherent states $\mathcal{C}$.  We will provide a physical interpretation of each set as wavefunctions of finite number of particles, generators of the classical algebra of observables and generators of the quantization mapping respectively in \cite{alonsoGeometricFlavours2024}. We claim that they are dense subsets of the whole space, for a proof of this statement in each case we refer to general arguments made at the beginning of  \cite{hidaWhiteNoise1993}. Each of those families are easily and unambiguously defined thanks to the properties of nuclear spaces and each of them show different properties when combined with the $S$ and $T$ transforms introduced in the previous section. 

\subsubsection{The algebra of Polynomials}
\label{sssec:Polynomials}

 The most important dense subset in Gaussian analysis are polynomials.  The fact tha $\mathcal{N}'_{\mathbb{C}}$ is a nuclear space implies that the tensor product $ {\mathcal{N}_\mathbb{C}'}^{\otimes k}$ is unambiguously defined as well as $\mathcal{N}_{\mathbb{C}}^{\otimes n}$ for $k,n$ positive integers. With this remark we define the algebra of  polynomials $\mathcal{P}(\mathcal{N}'_\mathbb{C})$, with coefficients in $\mathcal{N}_{\mathbb{C}}$, as the set of polynomials of arbitrary degrees in which  a polynomial of degrees $(n,m)$ is  written as

 \begin{equation}
     \label{eq:polynomialscompl} 
     P_{nm}(\phi^{\mathbf{x}},\bar\phi^{\mathbf{x}})= \sum_{k=0}^n\sum_{l=0}^m C^{kl}_{\vec{x}_k,\vec{y}_l}(\phi^k)^{\vec{x}_k}(\bar\phi^l)^{\vec{y}_l},
 \end{equation}
  where $C^{kl}_{\vec{\mathbf{x}}_k,\vec{\mathbf{y}}_l}\in \mathcal{N}_{\mathbb{C}}^{\otimes k+l}$ and $(\phi^k)^{\vec{\mathbf{x}}_k}\in {\mathcal{N}_\mathbb{C}'}^{\otimes k}$, $(\bar\phi^l)^{\vec{\mathbf{x}}_l}\in {\mathcal{{N}_\mathbb{C}'}^*}^{\otimes l}$ are the tensor product of copies of the same distribution.  Here we introduce the notation $\alpha_{\vec{x}_n}:=\alpha_{x_1,\cdots,x_n}$.
  To make sense of this expression is enough to consider a coefficient $C^{kl}_{\vec{\mathbf{x}}_k,\vec{\mathbf{y}}_l}=\chi^1_{x_1}\cdots\chi^k_{x_k}{\bar\rho}^1_{y_1}\cdots{\bar\rho}^l_{y_l}$ because these coefficients generate the whole space. Then, using the natural dual pairing of $\mathcal{N}_{\mathbb{C}}$ with its dual expressed as $\langle \chi,\phi \rangle= \chi_{x}\phi^x$, we get 
  $$
  C^{kl}_{\vec{x}_k,\vec{y}_l}(\phi^k)^{\vec{x}_k}
  (\bar\phi^l)^{\vec{y}_l}=\prod_{i=1}^k 
  \langle \chi^i,\phi \rangle
  \prod_{j=1}^l
   \overline{\langle \rho^j,\phi  \rangle}
 $$

 Thus polynomials defined in \eqref{eq:polynomialscompl}  have pointwise meaning because each dual pairing is finite. Nonetheless, there is still an ambiguity  in this definition  because the monomials are contractions of  $n$ or $m$ products of the same distribution. This fact implies there are different coefficients that provide the same polynomial. To kill this ambiguity we introduce the symmetrization operator 
    \begin{equation}
    \label{eq:symmetrization}
    \alpha_{(x_1,\cdots,x_n)}:=\frac{1}{n!}\sum_{\pi\in S_n}\alpha_{x_{\pi(1)},\cdots,x_{\pi(n)}}
    \end{equation}
    Here $S_n$ is the space of permutations of $n$ elements. Then for polynomials to be unambiguously defined through its coefficients we ask for them to fulfill $C^{kl}_{(\vec{\mathbf{x}}_k),(\vec{\mathbf{y}}_l)}=C^{kl}_{\vec{\mathbf{x}}_k,\vec{\mathbf{y}}_l}$ and we call them bosonic coefficients. Among polynomials we can find an orthonormal basis of the whole space and then describe it efficiently through its bosonic coefficients in an appropriate way using the $S$ transform. This feature is so rich and enlightening that it deserves its own separate analysis that will be the subject of study of \autoref{ssec:WienerIto}. 
 
\subsubsection{ The algebra of trigonomeric exponentials: The nature of Reproducing Kernel Hilbert space}

 Polynomials do not exhaust the classes of dense subsets that provide valuable information on the structure of $L^2(\mathcal{N}_\mathbb{C}, D\mu_c)$. In this section we turn our attention to the  algebra of trigonometric exponentials $\mathcal{T}(\mathcal{N}'_\mathbb{C})$, which can be seen  to be a dense subset of $L^2(\mathcal{N}_\mathbb{C}, D\mu_c)$. They are defined as   
\begin{equation}
\label{eq:trigExponentials}
\mathcal{T}(\mathcal{N}'_\mathbb{C})=\left\{\textrm{ Finite linear combinations of }  \mathcal{E}_{\chi}=\exp\left(i\chi_x\phi^x+i\overline{\chi_x\phi^x}\right)\right\}
\end{equation}
where $\chi_{\mathbf{x}}\in \mathcal{N}_{\mathbb{C}}$. Below we perform an analysis from the point of view of a basis of trigonometric exponentials.  If we consider the  $T_{\mu_c}$ transform of \autoref{def:TComplextransform}, its action over $\mathcal{T}(\mathcal{N}'_\mathbb{C})$ is completely  determined by linearity from  
\begin{equation}
\label{eq:TtransformTrig}
T_{\mu_c}[\mathcal{E}_{\chi}](\rho_\mathbf{x},\bar{\rho}_\mathbf{y})=C_c(\rho_{\mathbf{x}}+\chi_{\mathbf{x}},\overline{\rho_{\mathbf{y}}+\chi_{\mathbf{y}}}).
\end{equation}
Here $C_c$ is the characteristic functional given by \eqref{eq:complexCharacteristicFunctional}. Notice that in this equation $\rho_\mathbf{x}$ plays the role of a variable while $\chi_\mathbf{x}$ is a parameter. Hence, we interpret this expression as a function of $\rho_\mathbf{x}$ and extend the domain of  $T_{\mu_c}[\mathcal{E}_{\chi}](\rho_\mathbf{x},\bar{\rho}_\mathbf{y})$ to make sense in the whole Cameron-Martin Hilbert space  $\rho_{\mathbf{x}}\in \mathcal{H}_\Delta$. The next step is to construct a reproducing kernel Hilbert space $\mathcal{R}$. In order to do so, we firstly introduce the kernel functional defined as

\begin{equation}
    \label{eq:kernelforRKHS}
    \begin{array}{cccc}
    RK(\cdot,\cdot) :&\mathcal{H}_\Delta \times \mathcal{H}_\Delta & \longrightarrow & \mathbb{R} \\ 
     & (\rho_\mathbf{x},\chi_\mathbf{y}) &\longmapsto & C_c(\chi_{\mathbf{x}}-\rho_{\mathbf{x}},\overline{\chi_{\mathbf{y}}-\rho_{\mathbf{y}}}).
    \end{array}
    \end{equation}

    Using this definition, $\mathcal{R}$ is defined from the space of linear combinations of  functions (of $\rho,\bar\rho$) spanned by the r.h.s of \eqref{eq:TtransformTrig} as its completion with a sesquilinear scalar product $\langle,\rangle_{\mathcal{R}}$. It is enough to define this scalar product through its action over trigonometric exponentials under $T_{\mu_c}$ as

\begin{equation}
\label{eq:scalarProductCoherent}
\langle T_{\mu_c}[\mathcal{E}_{\rho}] , T_{\mu_c}[\mathcal{E}_{\chi}]\rangle_\mathcal{R}=RK(\rho_{\mathbf{x}},\chi_{\mathbf{x}}).
\end{equation}

A key factor in this construction is that $RK(\rho_{\mathbf{x}},\bullet )\in \mathcal{R}$ and for every $F\in \mathcal{R}$ we have  $\langle RK(\rho_{\mathbf{x}}, \bullet  ) , F(\bullet)\rangle_\mathcal{R}= F(\rho_{\mathbf{x}})$, this is, point evaluation is continuous in $\mathcal{R}$ which is the definition of reproducing Kernel Hilbert space.  

The remaining part of the discussion is to show that $T_{\mu_c}:L^2\big({\mathcal{N}_{\mathbb{C}}'},D\mu_c\big) \to \mathcal{R}$ is an isometric isomorphism of Hilbert spaces. This easily follows from 

\begin{equation}
\label{eq:ISometry}
\int_{\mathcal{N}'_\mathbb{C}}D\mu_c\overline{ \mathcal{E}_{\chi}} \mathcal{E}_{\rho}
=\langle T_{\mu_c}[\mathcal{E}_{\chi}] , T_{\mu_c}[\mathcal{E}_{\rho}]\rangle_\mathcal{R}.
\end{equation}

From this perspective we may ask in which sense point evaluation is allowed for elements in $L^2\big({\mathcal{N}_{\mathbb{C}}'},D\mu_c\big)$. It is important to notice that the domain of the functions that form this space is much larger than that of the functions of $\mathcal{R}$, given by $\mathcal{H}_\Delta$. Therefore there is no hope for the full notion of reproducing kernel  to be found in the former as we did in the latter, and as such, there is no notion of point evaluation to be found. The issue of point evaluation in this case is similar to the issue of defining a Dirac 
delta in finite dimensional domains. Similarly to that case, there is a notion of evaluation that makes sense in a distributional sense as we will see in \autoref{ex:repKernel} below.  At the level of Hilbert space, the family of functions considered in the next section provides tools to delve further in this question.

\subsubsection{The algebra of coherent states: Subspace of Holomorphic functions}

In this case we combine features from the two families considered above. 

Extending the algebra of polynomials to the limit of infinite degree, there is a natural definition of holomorphic function.       
\begin{equation}
        \label{eq:chaosdecompositionhol}
        \Psi_{Hol}(\phi^{\mathbf{x}})=\sum_{n=0}^\infty\psi^{n}_{\vec{x}_n}(\phi^n)^{\vec{x}_n}.
    \end{equation}

    Here we can allow the bosonic coefficients $\psi^{n}_{\vec{x}_n}$ to be elements of a larger space $\mathcal{H}_\Delta^{\otimes n}$.
    The  price to pay in that extension is that \eqref{eq:chaosdecompositionhol} is no longer defined pointwise. Instead, we can regard it as an element of $\Psi_{Hol}\in L^2_{Hol}\big({\mathcal{N}_{\mathbb{C}}'},D\mu_c\big) $ with the pertinent convergence restrictions.
    This is a general feature of the Wiener-Ito decomposition theorem that will be treated in detail in \autoref{ssec:WienerIto}. 
    We notice that, following this reasoning,  there is a decomposition in holomorphic and antiholomorphic functions as follows

\begin{equation}
    \label{eq:ComplexHilbertspacedecomp}
    L^2\big({\mathcal{N}_{\mathbb{C}}'},D\mu_c\big)= L^2_{Hol}\big({\mathcal{N}_{\mathbb{C}}'}, D\mu_c\big)\bar{\otimes}\overline{ L^2_{{Hol}}\big({\mathcal{N}_{\mathbb{C}}'}, D\mu_c\big)}
    \end{equation}
    where $\bar{\otimes}$ represents the topological tensor product completed with the natural Hilbert topology.

Let us  particularize our analysis to the subspace $L^2_{Hol}\big({\mathcal{N}_{\mathbb{C}}'},D\mu_c\big)$. Besides holomorphic polynomials, this space  contains another dense subspace of interest. We introduce the algebra of functions generated by (Holomorphic) coherent states $\mathcal{K}_{\chi}(\phi^{\mathbf{y}})=e^{\chi_x\phi^x}$ and we denote it 
\begin{equation}
    \label{eq:CoherentStates}
    \mathcal{C}_{Hol}(\mathcal{N}'_\mathbb{C})=\left\{\textrm{ Finite linear combinations of }  \mathcal{K}_{\chi}=\exp\left(\chi_x\phi^x\right)\right\}
    \end{equation}

  Using this subset calculus  can be easily  implemented on the whole space by just considering a few rules.  First, notice that
the scalar product of two coherent states corresponds to
\begin{equation*}
    \langle \mathcal{K}_{\chi},\mathcal{K}_{\rho}\rangle = e^{{\bar{\chi}_x\Delta^{xy}\rho_y}},
\end{equation*}
from which, by linearity,  follows the following reproducing property,  let $\Psi\in L^2_{Hol}\big({\mathcal{N}_{\mathbb{C}}'},D\mu_c\big)$, then 
\begin{equation*}
    \langle\Psi, \mathcal{K}_{\chi}\rangle= \Psi(\chi_{x}\Delta^{x \mathbf{y}}). \\
\end{equation*}
  Recall that  we extended \eqref{eq:chaosdecompositionhol} to have coefficients in $\mathcal{H}_\Delta^{\otimes n}$. Thus, we conclude that the evaluation above provides a finite result and $ \mathcal{K}_{\chi}(\phi^{\mathbf{y}})\in L^2_{Hol}\big({\mathcal{N}_{\mathbb{C}}'},D\mu_c\big)$. Those coherent states can be extended to be parametrized by elements $\rho^{\mathbf{x}}\in \mathcal{H}_\Delta$ regarded as distributions with $ \mathcal{K}_\rho(\phi^{\mathbf{y}})=e^{\bar\rho^xK_{xy}\phi^x}$. The duality, adapted to the conventions \eqref{eq:distHolomorph}, is expressed as in option b) of the previous section in terms of the inverse of the covariance, $K_{\mathbf{xy}}$, and therefore, lacking any pointwise meaning.

 This subspace possesses good analytical properties similar to those of holomorphic functions in finite dimensions. We can think about it as the Hardy space $H^2(\mathcal{H}_\Delta)\subset Hol(\mathcal{H}_\Delta)$  of square integrable holomorphic functions with domain an infinite dimensional Hilbert space $\mathcal{H}_\Delta$ whose orthogonal directions are provided by the covariance $\Delta^{\mathbf{xy}}$.
 There is a subtlety  with the sense of holomorphic function  that needs to be treated with care tough.  In order to do that, let us first introduce a more general notion of holomorphic function that the one introduced above: 
\begin{definition}
    \label{def:Holomorphic}\textbf{(Holomorphic functions on locally convex spaces)}
    Let $\Psi:\mathcal{E}\to\mathbb{C}$ be a function defined on an open set $U\subset \mathcal{E}$ where $\mathcal{E}$ is a locally convex space. In this work we take $\mathcal{E}= \mathcal{H}_\Delta, \mathcal{N}_{\mathbb{C}}$ or $\mathcal{N}_{\mathbb{C}}'$.  $\Psi$ is said to be Gâteaux-Holomorphic on $U$ if for all $\phi^\mathbf{x}_0,\phi^\mathbf{x}\in U$ the function
    \begin{equation}
    \label{eq:complexSubfunctionGateaux}
    \begin{array}{rccc}
        \Psi_{\phi^\mathbf{x}_0,\phi^\mathbf{x}}:&\mathbb{C}&\longrightarrow&\mathbb{C}\\
         & z & \longmapsto & \Psi(\phi^\mathbf{x}_0+z\phi^\mathbf{x})
    \end{array}
    \end{equation} 
    is holomorphic. 

    As it usually happens in infinite dimensions, continuity does not follow from differentiability and  it must be explicitly requested \cite{dineenComplexAnalysis1981}. Thus we will say that a function  $\Psi$ is  holomorphic on $U$ if it is Gâteaux-Holomorphic and locally bounded. We will denote then the     space of holomorphic functions as $Hol(\mathcal{E})$, assuming it is endowed with the topology of convergence on bounded subsets.
    \end{definition}

   Notice here that the definition of holomorphic function  requires the pointwise meaning of $\Psi$.  This is not true if the decomposition in \eqref{eq:chaosdecompositionhol} is allowed to have coefficients in $\mathcal{H}_\Delta^{\otimes n}$. For $\Psi\in L^2_{Hol}\big({\mathcal{N}_{\mathbb{C}}'},D\mu_c\big)$ and arbitrary elements of $\phi^{\mathbf{x}}\in \mathcal{N}'$ we have $\Psi(\phi^{\mathbf{x}})$ does not necessarily converge. However, consider this notion of  $Hol(\mathcal{H}_\Delta)$ over the Hilbert space $\mathcal{H}_\Delta$.  The  convergence of \eqref{eq:chaosdecompositionhol} is ensured if we restrict $\phi^{\mathbf{x}}\in \mathcal{H}_\Delta\subset \mathcal{N}'$. The content of this statement can be made precise using coherent states as follows. Let us consider  the map  
\begin{equation}
\label{eq:injectionHolomorphic}
\langle\cdot, \mathcal{K}_{\rho}\rangle: L^2_{Hol}\big({\mathcal{N}_{\mathbb{C}}'},D\mu_c\big) \to H^2(\mathcal{H}_\Delta)\subset Hol(\mathcal{H}_\Delta) 
\end{equation}
which is an injective continuous linear mapping \cite{oecklHolomorphicQuantization2012}, and $H^2(\mathcal{H}_\Delta)$ stands for its image. Moreover, the reproducing kernel 
\begin{equation}
\label{eq:kernel}
RK_{Hol}(\chi_{\mathbf{x}},\rho_\mathbf{y})=e^{\overline{\chi}_x\Delta^{xy}\rho_y}
\end{equation}
turns  $(H^2(\mathcal{H}_\Delta),RK_{Hol})$ into a reproducing kernel Hilbert space   known as the Hardy space of square integrable functions   . Furthermore, \eqref{eq:injectionHolomorphic} establishes an isometric isomorphism between $L^2_{Hol}\big({\mathcal{N}_{\mathbb{C}}'},D\mu_c\big)$ and $H^2(\mathcal{H}_\Delta)$.

A similar discussion can be carried out for the anti-holomorphic functions but in that case the isomorphism is realized by 

\begin{equation}
\label{eq:injectionAntiHolomorphic}
\langle \cdot ,\overline{\mathcal{K}_{\rho}}\rangle: L^2_{\overline{Hol}}\big({\mathcal{N}_{\mathbb{C}}'},D\mu_c\big) \to \overline{H^2}(\mathcal{H}_\Delta)
\end{equation}

Finally, the $T_{\mu_c}$ transform of \autoref{def:TComplextransform}  can be written with these objects as  $T_{\mu_c}=\langle \overline{\mathcal{K}_{i\rho}}\otimes \mathcal{K}_{-i\rho} , \cdot\rangle$. Therefore, there is an isometric isomorphism of reproducing kernel Hilbert spaces such that

\begin{equation}
\label{eq:isomorphism}
    H^2(\mathcal{H}_\Delta)\otimes \overline{H^2}(\mathcal{H}_\Delta) \cong \mathcal{R},
\end{equation}
where $\mathcal{R}$ is the reproducing kernel Hilbert space described in the previous section.

\subsubsection{ Some physical implications}

So far we have seen how the dense subsets of trigonometric exponentials and coherent states unveil several properties of the space $L^2\big({\mathcal{N}_{\mathbb{C}}'},D\mu_c\big)$. In \cite{alonsoGeometricFlavours2024} this kind of  space will play a twofold role. On one hand a Gaussian Hilbert space of complex domain is a prequantum Hilbert space that is a  basic ingredient in the geometric quantization procedure. On the other a similar Gaussian space (with different covariance), represents the class of classical functions that can be quantized under Weyl quantization.

Among  the properties above the most important one for geometric quantization is the decomposition in holomorphic and antiholomorphic parts \eqref{eq:ComplexHilbertspacedecomp}. This is because, as we will see in \cite{alonsoGeometricFlavours2024}, the Hilbert space of pure states in the quantum field theory can be represented in several equivalent ways from which the holomorphic representation selects the subspace  $L^2_{Hol}\big({\mathcal{N}_{\mathbb{C}}'},D\mu_c\big)$.
Finally, the fact that \eqref{eq:injectionHolomorphic} provides a isomorphism of Hilbert spaces with  $H^2(\mathcal{H}_\Delta)$, a reproducing kernel Hilbert space with a well behaved domain $\mathcal{H}_\Delta$, ensures nice analytical properties for  $L^2_{Hol}\big({\mathcal{N}_{\mathbb{C}}'},D\mu_c\big)$ that allow for an   efficient calculus and study of field theory keeping technicalities under control.

For the case of Weyl quantization in \cite{alonsoGeometricFlavours2024} we show that it is enough to understand the quantization of trigonometric exponentials \eqref{eq:trigExponentials}  as plane waves of the classical field theory  for the machinery of reproducing kernels  to quantize them according to the Weyl or the canonical commutation relations. 

The features unveiled by the family of polynomials are much richer and we devote the next section to them. We advance that the Wiener-Ito decomposition theorem will be the tool that provides the particle interpretation of the QFT.

\subsection{Modelling pure  quantum field states I: Wiener Ito decomposition}
\label{ssec:WienerIto}
In this section we will consider the mathematical tools required to represent the concept of pure state of a quantum field theory and  their main properties. In particular the relation with the concept of physical particles.

For a scalar filed we can consider two types of unitary equivalent representations.  On one hand, the Hilbert space of the Schrödinger wave functional representation $\Psi[\varphi]$ of quantum states is given by $L^2(\mathcal{N}', D\mu)$, for a given Gaussian measure $\mu$.  This Gaussian measure is often interpreted, in the physics literature, as part of the vacuum state of the theory \cite{longSchrodingerWave1998,corichiSchrodingerFock2004,hofmannNonGaussianGroundstate2017}. On the other hand, for the holomorphic representation $\Psi[\phi]$, it is $L^2_{Hol}(\mathcal{N}'_\mathbb{C}, D\mu_c)$, for its respective complex Gaussian measure. \textit{Id est}, a pure state is to be thought as a function with domain in the DNF space $\mathcal{N}'$ or its complexification $\mathcal{N}'_\mathbb{C}$. For further details on the physical meaning of each representation we refer to the second part of this series  \cite{alonsoGeometricFlavours2024}.  

It is common knowledge that the Hilbert space of a scalar field theory is a bosonic Fock space. This  means that it can be expressed as $\oplus_{n=0}^\infty \mathcal{H}^{n\hat\oplus}_{\Delta}$ where $\oplus$ represents the orthogonal direct sum,  $\mathcal{H}_{\Delta}$ represents the one particle states of the theory and $n\hat\otimes$ means the symmetric tensor product of $n$ copies of the Hilbert space. In this context, the notion of particle in QFT is defined as the eigenstates of the  particle number operator.   The precise relation between the Fock space and the $L^2$ Hilbert spaces is given by the Wiener-Ito decomposition in homogeneous chaos, which we discuss briefly below. In the same lines, we can also identify common tools of Gaussian analysis as Malliavin derivatives and Skorokhod integrals with creation and annihilation operators defining in the process the notion of charge and particle number.

\subsubsection{The Segal isomorphism: Wiener-Ito chaos decomposition and the Fock space}
The classical phase space $\mathcal{M}_C$ can be considered as the set of complex test functions $\mathcal{N}_\mathbb{C}$ once we introduce a suitable complex structure on the set $\mathcal{N}\times \mathcal{N}$ (containing the fields and their momenta). We will start dealing with the complex case and then particularize our result to the real case. If the Gel'fand triple discussed above is considered with respect to complex spaces, we can build the triple

\begin{equation}
\label{eq:complexgelfand}
\mathcal{M}_C\sim \mathcal{N}_\mathbb{C}\subset \mathcal{H}_\Delta\subset \mathcal{N}'_\mathbb{C},
\end{equation}
 where $\mathcal{H}_\Delta$ represents the, now complex, Cameron-Martin Hilbert space. By using this triple, we can consider $\mathcal{M}_C$ as a linear subspace of $\mathcal{N}'_\mathbb{C}$. This means that we will think at our classical fields and their momenta as distributions, dual to the nuclear space of complex smooth functions of compact support on $\Sigma$.

Now, we move to the level of second quantization. We therefore deal with the Hilbert space of pure quantum states which is  a subspace of $L^2(\mathcal{N}'_\mathbb{C}, D\mu_c)$. This  larger space is called the prequantum Hilbert space. For further insight in the role of these objects in quantization see the second part of this series \cite{alonsoGeometricFlavours2024}.   Let us consider now an orthogonal decomposition of the algebra of polynomials $\mathcal{P}(\mathcal{N}'_{\mathbb{C}})$ introduced above in  \eqref{eq:polynomialscompl}. We will see that this construction allows us to define an isomorphism of this space of quantum states with the (bosonic) Fock space constructed from the  (Cameron-Martin) Hilbert space $\mathcal{H}_\Delta$. Thus, we define  $\mathcal{W}_{\mu_c}^{:(n,m):}$ and denote it the space of homogeneous complex chaos of degrees $(n,m)$ as
\begin{equation}
\label{eq:wickspaces}
\mathcal{W}_{\mu_c}^{:(n,m):}=\overline{\mathcal{P}^{(n,m)}}\cap \overline{\mathcal{P}^{{(n-1,m)}}}^\perp\cap \overline{\mathcal{P}^{{(n,m-1)}}}^\perp,
\end{equation}
where $\overline{\mathcal{P}^{(n,m)}}$ represents the completion (in $L^2(\mathcal{N}'_\mathbb{C}, D\mu_c)$) of the set of polynomials of degrees $(n,m)$  and $\overline{\mathcal{P}^{(n-1,m)}}^\perp$ the orthogonal complement of the completion of the set of polynomials of degrees $(n-1,m)$. The sequence of spaces  $\mathcal{W}_\mu^{:(n,m):}$  contain a basis of orthogonal polynomials for the Hilbert space $L^2(\mathcal{N}'_\mathbb{C}, D\mu_c)$. In addition, these spaces are   generated by complex Wick monomials of the Gaussian measure $\mu_c$ (or, equivalently, of the covariance $\Delta^{\mathbf{xy}}$)   generalizing the definition of complex Hermite polynomials to this more abstract setting \cite{hidaBrownianMotion1980} as follows

\begin{definition}
    \label{def:ComplexWickMonomials}{\bf (Complex Wick monomials)} 
    Let  $\Delta^{\mathbf{xy}}$ be the covariance of the Gaussian measure $\mu$.  
    We will use $\mathcal{W}_\Delta(\phi^n\bar{\phi}^m)^{\vec{\mathbf{x}}\vec{\mathbf{y}}}$, with $\vec{x}=(x_i)_{i=1}^n, \vec{y}=(y_i)_{i=1}^m$, to refer to the distribution belonging to the symmetric tensor product of $n$ copies on $\mathcal{N}'_\mathbb{C}$ and the symmetric product of $m$ copies on ${\mathcal{N}'_\mathbb{C}}^*$ i.e. $\mathcal{W}_\Delta(\phi^n\bar{\phi}^m)^{\vec{\mathbf{x}}\vec{\mathbf{y}}}\in{\mathcal{N}'_\mathbb{C}}^{\hat{\otimes}n}\otimes {\mathcal{N}'_\mathbb{C}}^{*\hat{\otimes}m}$ such that
    \begin{equation}
    \label{eq:ComplexWienerIto}
    \mathcal{W}_\Delta(\phi^n\bar{\phi}^m)^{\vec{x}\vec{y}}:=\prod_{i=1}^n\prod_{j=1}^m\frac{\delta}{\delta \rho_{x_i}}\frac{\delta}{\delta \bar\rho_{y_j}}\exp\left[{\overline{\rho_z\phi^z}+\rho_z\phi^z-\bar\rho_u\Delta^{uv}\rho_v}\right]\bigg\lvert_{\rho_{\mathbf{x}}=0}\;.
    \end{equation}
    We will refer to this distribution as \textit{complex Wick monomial} of degrees $(n,m)$ with respect to the covariance $\Delta^{\mathbf{xy}}$.
\end{definition}

 With this set of polynomials, we can write any state in the Hilbert space $L^2(\mathcal{N}_\mathbb{C}, D\mu_c) $ as a linear combination 
    \begin{equation}
    \label{eq:ComplexWienerItodec}
    \Psi(\phi^{\mathbf{x}},\bar\phi^\mathbf{y})=\sum_{n,m=0}^\infty \psi^{(n,\bar{m})}_{\vec{x}_n\vec{y}_m} \mathcal{W}_\Delta(\phi^n\bar{\phi}^m)^{\vec{x}_n\vec{y}_m}
    \end{equation}
    We call this decomposition Wiener-Ito decomposition in homogeneous chaos.
    Notice that the set of polynomials can be decomposed in holomorphic and anti-holomorphic parts 
    \begin{equation}
    \label{eq:ComplexPolinomials}
    \mathcal{P}_{Hol}(\mathcal{N}'_{\mathbb{C}})\otimes \mathcal{P}_{\overline{Hol}}({\mathcal{N}_{\mathbb{C}}'}^*)\subset L^2\big({\mathcal{N}_{\mathbb{C}}'},D\mu_c\big)= L^2_{Hol}\big({\mathcal{N}_{\mathbb{C}}'}, D\mu_c\big)\bar{\otimes} L^2_{\overline{Hol}}\big({\mathcal{N}_{\mathbb{C}}'}, D\mu_c\big).
    \end{equation}
    This decomposition provides an exceptional tool for Gaussian analysis, in terms of the coefficients on the Wiener-Ito decomposition we can write the scalar product as:
    \begin{align}
    \label{eq:complexchaosDecomposition}
    \int_{\mathcal{N}'_\mathbb{C}}D\mu_c(\phi^\mathbf{x})\overline{\Phi}(\phi^\mathbf{x})\Psi(\phi^\mathbf{x})
    &= \nonumber \\ \sum_{n,m=0}^\infty n!m!&\ \overline{\varphi^{(n,\bar{m})}}_{\vec{u}_n\vec{v}_m}\big[(\Delta^*)^n\big]^{\vec{u}_n\vec{x}_n} \big[\Delta^m\big]^{\vec{v}_m\vec{y}_m} \psi^{(n,\bar{m})}_{\vec{x}_n\vec{y}_m}
    =\nonumber \\ &
    \sum_{n,m=0}^\infty n!m! \big\langle\varphi^{(n,\bar{m})}_{\vec{\mathbf{x}}_n\vec{\mathbf{y}}_m},\psi^{(n,\bar{m})}_{\vec{\mathbf{x}}_n\vec{\mathbf{y}}_m}\big\rangle_{\mathcal{H}^{(n,\bar{m})}_\Delta}.
    \end{align}
   Here $\varphi^{(n,\bar{m})}_{{\vec{\mathbf{x}}_n\vec{\mathbf{y}}_m}}$ represent the coefficients of the chaos decomposition of $\Phi$.Therefore, once weighted with $\sqrt{n!m!}$, the coefficients  of the Wiener-Ito decomposition belong to the Bosonic Fock space $\Gamma\mathcal{H}_\Delta\otimes{\Gamma\mathcal{H}_\Delta^*}$ where ${\Gamma\mathcal{H}_\Delta^*}$ represents the complex conjugate of the Fock space.  To find the coefficients of the decomposition is useful to use the $S_{\mu}$ transform as
    \begin{equation}
    \label{eq:wienerItoSeagalwithTtransform}
    \psi^{(n,\bar{m})}_{\vec{x}_n\vec{y}_m}=\frac{1}{n!m!}\prod_{i=1}^n\prod_{j=1}^m K_{x_iu_i}K^*_{y_jv_j}\frac{\partial}{\partial \bar\rho_{u_i}}\frac{\partial}{\partial \rho_{v_j}}  {S_\mu[\Psi](\rho_\mathbf{x},\bar{\rho}_\mathbf{x})} \bigg\lvert_{\rho_\mathbf{x}=0}
    \end{equation}
    Let us recall that $K_{\mathbf{xy}}$ is the  (weak) inverse of the covariance $\Delta^{\mathbf{xy}}$.
     We see that the decomposition of \eqref{eq:complexchaosDecomposition} is such that
     
     \begin{equation}
        \label{eq:holantiholdec}
        L^2\big({\mathcal{N}_{\mathbb{C}}'}, D\mu_c\big)=L^2_{Hol}\big({\mathcal{N}_{\mathbb{C}}'}, D\mu_c\big)\bar{\otimes} L^2_{\overline{Hol}}\big({\mathcal{N}_{\mathbb{C}}'}, D\mu_c\big)=\bigoplus _{n=0}^\infty\left( \bigoplus _{m=0}^n \mathcal{H}_\Delta^{(n-m,\bar{m})} \right).
    \end{equation}
    
    This leads us to build an isomorphism between the Hilbert space and the Fock space which associates each state in $L^2({\mathcal{N}_{\mathbb{C}}'}, D\mu_c)$ and its coordinates in the linear decomposition:

\begin{definition}
    \label{def:SegalIso}{\bf (Segal isomorphism)} 
    We will call Segal isomorphism to the mapping

    \begin{equation}
\label{eq:Segal}
\begin{array}{cccc}
    \mathcal{I} :&  L^2\big({\mathcal{N}_{\mathbb{C}}'}, D\mu_c\big) & \longrightarrow &  \Gamma\mathcal{H}_\Delta\otimes{\Gamma\mathcal{H}_\Delta^*}
 \\ 
 &  \Psi &\longmapsto &  \mathcal{I}(\Psi)=\left (\sqrt{n!m!}\psi^{(n,\bar{m})}_{\vec x_n \vec y_m} \right )_{n,m=0}^\infty 
\end{array}
\end{equation}
which defines a unitary isomorphism of Hilbert spaces considering the scalar product defined in \eqref{eq:complexchaosDecomposition}. In this decomposition the Fock space is further decomposed into the orthogonal direct sum $\Gamma\mathcal{H}_\Delta\otimes {\Gamma\mathcal{H}_\Delta^*}
=\bigoplus_{n=0}^\infty \left( \bigoplus_{m=0}^n \mathcal{H}_\Delta^{(n-m,\bar{m})} \right)$.
\end{definition} 

Two cases are particularly relevant for us:

\paragraph*{The holomorphic representation}
    The holomorphic representation is described by the holomorphic subspace $L^2_{Hol}\big({\mathcal{N}_{\mathbb{C}}'}, D\mu_c\big)$. Their Wick monomials are particularly simple $\mathcal{W}_\Delta(\phi^n\bar{\phi}^0)^{\vec{\mathbf{x}}}= (\phi^n)^{\vec{\mathbf{x}}}$. This fact simplifies remarkably computations in because holomorphic functions have a Wiener-Ito  decomposition given by 
    \begin{equation}
    \label{eq:holchaosdecomposition}
    \Psi_{Hol}(\phi^{\mathbf{x}})=\sum_{n=0}^\infty\psi^{n}_{\vec{x}_n}(\phi^n)^{\vec{x}_n}
    \end{equation}
     independently of the Gaussian measure considered. Besides, the coefficients of each regular monomial belong to orthonormal subspaces. 

     \paragraph*{The Schrödinger representation}
    For the real case the subspace is $L^2\big({\mathcal{N}'}, D\mu\big)$. In this case we also have a Wiener-Ito decomposition.  We work out the decomposition from the real $S_{\mu}$ transform. Therefore Wick monomials are given by

    \begin{equation}
    \label{eq:wickMonomialsGenerator}
    :\varphi^n:\lvert_\Delta^{\vec{\mathbf{x}}}=
    \prod_{i=1}^n\frac{\delta}{\delta \xi_{y_i}}
    \exp{\big(\xi_x\varphi^x-\frac12\xi_x\Delta^{xy}\xi_y\big)} \lvert_{\xi_{\mathbf{x}}=0}
    \end{equation}
    or recursively 
\begin{definition}
    \label{def:wickMonomials}{\bf (Wick monomials)} 
    We denote $:\varphi^n:\lvert_\Delta^{\vec{\mathbf{x}}}=\ :\varphi^n:\lvert_\Delta^{\mathbf{x_1,\cdots,x_n}}$, and call Wick monomial of degree $n$ with respect to the covariance $\Delta^{\mathbf{xy}}$, the distribution belonging to the symmetric tensor product of $n$ copies on $\mathcal{N}'$ i.e. $:\varphi^n:\lvert_\Delta^{\vec{\mathbf{x}}}\in(\mathcal{N}')^{n\hat{\otimes}}$ such that
    \begin{align}
        :\varphi^0:\lvert_\Delta&=1\nonumber\\
        :\varphi^1:\lvert_\Delta^{\mathbf{x}}&=\varphi^{\mathbf{x}}\nonumber\\
        :\varphi^n:\lvert_\Delta^{\vec{\mathbf{x}},{\mathbf{x}}_{n-1},{\mathbf{x}}_n}&=
        :\varphi^{n-1}:\lvert_\Delta^{\vec{{\mathbf{x}}},{\mathbf{x}}_{n-1}}\varphi^{{\mathbf{x}}_n}-(n-1)\Delta^{\mathbf{x}_n(\mathbf{x}_{n-1}}:\varphi^{n-2}:\lvert_\Delta^{\vec{\mathbf{x}})}
        \label{eqn:wickPowers}
    \end{align}
    where the parenthesized indices are symmetrized according to the convention  \eqref{eq:symmetrization}. 
\end{definition}

  Then the Wiener-Ito decomposition is given by
    \begin{equation}
        \label{eq:realwienerIto}
        \Psi(\varphi^\mathbf{x})=\sum_{n=0}^\infty \psi^{(n)}_{\vec{x_n}}:\varphi^n:\lvert_\Delta^{\vec{{x}_n}}
    \end{equation}
 the coefficients belong to the weighted Bosonic Fock space obtained from $\Gamma\mathcal{H}_\Delta$ with weights $\sqrt{n!}$ and the scalar product is 
\begin{multline}
    \label{eq:complexchaosprod}
    \int_{\mathcal{N}}D\mu(\varphi^\mathbf{x})\overline{\Phi}(\varphi^\mathbf{x})\Psi(\varphi^\mathbf{x})
    = \\   \sum_{n=0}^\infty n!\ \overline{\phi^{(n)}}_{\vec{u}_n}\big[\Delta^n\big]^{\vec{u}_n\vec{x}_n}  \psi^{(n)}_{\vec{x}_n}
    = 
    \sum_{n,}^\infty n! \big\langle\phi^{(n)}_{\vec{\mathbf{x}}_n},\psi^{(n)}_{\vec{\mathbf{x}}_n}\big\rangle_{\mathcal{H}^{n}_\Delta}
    \end{multline}
    Here $\phi^{(n)}_{{\vec{\mathbf{x}}_n}}$ represent the coefficients of the chaos decomposition of $\Phi$.
    Notice that this recursive definition of Wick powers coincides with the usual prescription for normal ordering in quantum field theory where $\Delta^{\mathbf{xy}}$ resembles the propagator of the theory considered. However, that prescription is rooted in an ordering choice for creation and annihilation operators. This is not the case for this representation as they appear as generators of orthogonal polynomials. Moreover $\Delta^{\mathbf{xy}}$ only integrates over spatial coordinates and therefore is not a propagator.  Nonetheless,  this resemblance allows us to compute pointwise products of functions in terms of their chaos decompositions using the well known tool of Feynman diagrams \cite{moshayediQuantumField2019}.

    \subsubsection{The Seagal-Bargmann tranform: Isomorphisms between the real and holomorphic cases.}
    In the case where the covariance for the complex case, $\Delta^{\mathbf{xy}}$ is real and symmetric we can establish a unitary isomorphism between the subspace of holomorphic functions $L^2_{Hol}(\mathcal{N}'_{\mathbb{C}},D\mu_c)$ and  $L^2(\mathcal{N}',D\mu)$ just by performing an integration over the imaginary part of the coordinates:

    \begin{definition}
        \label{def:Segaltranf}{\bf (Segal-Bargmann transform)} We denote by Segal-Bargmann transform the unitary isomorphism
        $$
        \begin{array}{cccc}
        \mathcal{B} :&  L^2(\mathcal{N}',D\mu)  & \longrightarrow &  L_{Hol}^2(\mathcal{N}_{\mathbb{C}}',D\mu_c) 
        \end{array}
        $$
        where both $\mu,\mu_c$ are Gaussian measures with a real symmetric covariance $\Delta^{\mathbf{xy}}$, adapted in each case to its domain. To write the explicit formula for this isomorphism we take an arbitrary $\Psi_{Real}\in L^2(\mathcal{N}',D\mu)$ and denote $\mathcal{B}(\Psi_{Real})=\Psi_{Hol}$. Then it follows   $\int_{\mathcal{N}'_\mathbb{C}} D\mu_c(\phi^{\mathbf{x}})=\int_{\mathcal{N}'} D\mu(\varphi^{\mathbf{x}})\int_{\mathcal{N}'} D\mu(\pi^{\mathbf{x}})$. Let $\Psi_{Hol}$ be a holomorphic function with chaos decomposition given by \eqref{eq:holchaosdecomposition}, integrating over the imaginary degrees of freedom 
    \begin{equation*}
        \label{eq:schHolrelation}
        \begin{split}
        \int D\mu(\pi^{\mathbf{x}})\Psi_{Hol}(\phi^\mathbf{x})=
        \sum_{n=0}^\infty &\int D\mu(\pi^{\mathbf{x}}) \psi^{(n)}_{\vec{x}_n}{[(\varphi-
        i\pi)^n]^{\vec{x}_n}}=
        \\
        \sum_{n=0}^\infty  \psi^{(n)}_{\vec{x}_n}\frac{\partial^n}{\partial \xi^n_{\vec{x}_n}}\int D\mu(\pi^{\mathbf{x}})& e^{\xi_z(\varphi-i\pi)^z}\bigg\lvert_{\xi_{\mathbf{x}}=0}= 
        \\
        \sum_{n=0}^\infty  \psi^{(n)}_{\vec{x}_n}
        \frac{\partial^n}{\partial \xi^n_{\vec{x}_n}}\Bigg[&   e^{{\xi_z\varphi^z} } \int D\mu(\pi^{\mathbf{x}}) e^{-i{\xi_z\pi^z}}\Bigg]_{\xi_{\mathbf{x}}=0}= 
        \\
        \hspace{ 2 em}\sum_{n=0}^\infty \psi^{(n)}_{\vec{x}_n}\frac{\partial^n}{\partial \xi^n_{\vec{x}_n}}e^{\xi_x\varphi^x-\frac{\xi_u\Delta^{uv}\xi_v}{2}}\Big\lvert_{\xi_\mathbf{x}=0}&=
        \sum_{n=0}^\infty \psi^{(n)}_{\vec{x}_n}:\varphi^n:\lvert_{{\Delta}}^{\vec{x}_n}=
         \Psi_{Real}(\varphi^{\mathbf{x}})
        \end{split}
    \end{equation*}
    where $\xi_\mathbf{x}\in \mathcal{N}$ is real.   With this consideration,  $\mathcal{B}(\Psi_{Real})$ is computed obtaining the chaos decomposition of  $\Psi_{Real}$ and replacing Wick monomials by regular holomorphic monomials. 
    
    A similar isomorphism is obtained using the antiholomorphic subspace in the analogous form
    
    $$
    \overline{\mathcal{B}} :L^2(\mathcal{N}',D\mu)   \longrightarrow  L_{\overline{Hol}}^2(\mathcal{N}_{\mathbb{C}}',D\mu_c).
    $$
\end{definition}

These transforms play a relevant role in the construction of physical quantum field theories, as they allow for the relation of the holomorphic and  Schr\"odinger (or antiholomorphic and field-momentum)    functional representations of quantum states with respect to their appropriate Gaussian measures. See \cite{alonsoGeometricFlavours2024} for further details.

\subsubsection{Malliavin calculus: creation and annihilation operators and charges}
\label{sssec:MaliiavinCalculus}
The space $L^2(\mathcal{N}',D\mu)$ is endowed with some useful properties that may be used to model further relevant tools in QFT such as creation and annihilation operators, or the number operator. Let us briefly introduce them in the simpler case of real fields and  generalize them later to the case of complex ones.

As in the complex case a polynomial of  a real variable $\varphi^{\mb{x}}\in \mathcal{N}'$,  with coefficients in $\mathcal{N}_{\mathbb{C}}$ and degree $n$  can be written as
\begin{equation}
\label{eq:realpolynomials} 
P_n(\varphi)= C^0+C^1_x\varphi^x+C^{2}_{x_1x_2}(\varphi^2)^{x_1x_2}+\cdots+C^n_{\vec{x}_n}(\varphi^n)^{\vec{x_n}}
\end{equation}    
 The algebra of polynomials of a real variable is denoted by $\mathcal{P}(\mathcal{N}')$. This set is dense in $L^2(\mathcal{N}',D\mu)$ so we start introducing the operations of Malliavin calculus on $\mathcal{P}(\mathcal{N}')$ and then extend them to the whole space by linearity. On polynomials, the Fréchet derivative defined as in \cite{hidaWhiteNoise1993} has a particularly useful expression in terms of the operator $\partial_{\varphi^\mathbf{y}}=\partial_{\mathbf{y}}$ that can be defined over monomials as 
\begin{equation}
\label{eq:Fréchetderivative}
\partial_{\varphi^y} C_{x_1,\cdots,x_n}(\varphi^n)^{x_1,\cdots,x_n}=nC_{y,x_1,\cdots,x_{n-1}}(\varphi^{n-1})^{x_1,\cdots,x_{n-1}} .
\end{equation}
By linearity, this operator is extended  to $\mathcal{P}(\mathcal{N}')$ without further complications.

The Malliavin derivative is the extension of the Fréchet derivative, defined in \eqref{eq:Fréchetderivative}, to an operator 
\begin{equation}
\label{eq:malliavinDerivative}
\partial_\mathbf{x}: \mathbb{D}_\mu^{2,1}\longrightarrow \mathcal{H}_\Delta\otimes L^2(\mathcal{N}',D\mu)
\end{equation}
where the domain of the operator is
\begin{equation}
    \label{eq:MalliavinSovoleb}
    \mathbb{D}_\mu^{2,1}=\left\{\Psi(\varphi^\mathbf{x})=\sum_{n=0}^\infty \psi^{(n)}_{\vec{x_n}}:\varphi^n:\lvert_\Delta^{\vec{{x}_n}}\, \textrm{ s.t }\, \sum_{n=0}^\infty (1+n)\cdot n!\lVert\psi^{(n)}_{\vec{\mathbf{x}}_n}\rVert_{\Delta}^2< \infty \right\}
\end{equation}
and acts over $\Psi(\varphi^\mathbf{x})$ as

\begin{equation}
\label{eq:actionMalliavin}
\partial_x\Psi(\varphi^\mathbf{x})=\sum_{n=1}^\infty n\psi^{(n)}_{x,\vec{y}_{n-1}}:\varphi^{n-1}:\lvert_\Delta^{\vec{y}_{n-1}}
\end{equation}
The key feature of this derivative is that the resulting derivative makes sense as a $\mathcal{H}_\Delta$ function and does not have a pointwise meaning as  \eqref{eq:Fréchetderivative}. The Malliavin derivative can thus be considered as the generator of derivatives in directions of the Cameron-Martin Hilbert space. Furthermore, through the isomorphism $\mathcal{I}$, the contraction $\Delta^{xy}\partial_y$ is mapped to the annihilation operator on the Fock space $\Gamma(\mathcal{H}_\Delta)$. See \cite{alonsoGeometricFlavours2024} for its connection with the annihilation operator of QFT.

The adjoint operator of the Malliavin derivative is the Skorokhod integral $\partial^{*\mathbf{x}}$. It is an operator defined over $\textrm{Dom}(\partial^{*\mathbf{x}})\subset\mathcal{H}_\Delta\otimes L^2(\mathcal{N}',D\mu)$ 

\begin{equation}
    \label{eq:SkorokhodIntegral}
    \begin{array}{cccc}
        \partial^{*\mathbf{x}}:& \textrm{Dom}(\partial^{*\mathbf{x}}) &\longrightarrow &L^2(\mathcal{N}',D\mu)\\
         & \Psi_{\mathbf{x}}(\varphi^{\mathbf{y}}) & \longmapsto & \partial^{*{x}} \Psi_{x}(\varphi^{\mathbf{y}})
    \end{array}
\end{equation}
It is important to remark that the Skorokhod integral is able to integrate more elements than those obtained from \eqref{eq:actionMalliavin}. Indeed, a general element $\Psi_{\mathbf{x}}\in\mathcal{H}_\Delta\otimes L^2(\mathcal{N}',D\mu)$ admits a Wiener-Ito decomposition in which 
\begin{equation}
\label{eq:parametricWienerIto}
\sum_{n=0}^\infty\psi^{(n)}_{x,\vec{y}_n}:\varphi^n:\lvert_\Delta^{\vec{y}_n}\hspace{0.25 cm} \textrm{ in general } \hspace{0.25 cm}  \psi^{(n)}_{(x,\vec{y}_n)}\neq \psi^{(n)}_{x,\vec{y}_n}
\end{equation}
In other words $(\psi^{(n)}_{x,\vec{y}_n})_{n=0}^\infty$ is not necessarily symmetrized and hence it may not be an  element of the Bosonic Fock space $\Gamma(\mathcal{H}_\Delta)$. Nonetheless if  
$\tilde{\psi}^{(n+1)}_{\vec{\mathbf{z}}_{n+1}}=\psi^{(n)}_{(\mathbf{x},\vec{\mathbf{y}}_n)}$ is indeed the decomposition of an element of $\Gamma(\mathcal{H}_\Delta)$ then $\Psi_{\mathbf{x}}\in  \textrm{Dom}(\partial^{*\mathbf{x}})$ and 

\begin{equation}
\label{eq:actionSkorokhod}
\partial^{*{x}}\Psi_{{x}}=\sum_{n=1}^{\infty} \tilde{\psi}^{(n)}_{\vec{z}_n}:\varphi^n:\lvert_\Delta^{\vec{z}_n}
\end{equation}
This, for a large class of functions in $\textrm{Dom}(\partial^{\star\mathbf{x}})$, can be written in terms of the Malliavin derivative as 
\begin{equation}
\label{eq:creationOperator}
\partial^{*{x}}=-\Delta^{xy}\partial_y+\varphi^x
\end{equation} 

Notice that the Skorokhod integral is a generalization of the Ito integral, that coincides with it when the domain has dimension one. In this case, it  is identified with a time parameter, and we consider only $\Psi_x$ as a process adapted to the filtration of Brownian motion, see \cite{nunnoMalliavinCalculus2009}. See \cite{alonsoGeometricFlavours2024} for its connection with the creation operator of QFT.
It is even possible to extend these operations to the space of Hida distributions  $(\mathcal{N}_\mathbb{C})'$ introduced in the next section. For further details on Hida-Malliavin calculus,  we refer the interested reader to  \cite{nunnoMalliavinCalculus2009}.

Another operator that can be easily written in terms of the latter two is the particle number 

\begin{equation}
\label{eq:particleNumber}
N_{\Delta}=\partial^{*x}\partial_x=-\Delta^{xy}\partial_x\partial_y+\varphi^x\partial_x
\end{equation}

These notions can be easily extended to the holomorphic and antiholomorphic cases, which will represent matter-antimatter  models. 
\begin{gather}
\label{eq:holomorphicMalliavin}
\partial_{\phi^x}\Psi(\phi^{\mathbf{x}},\bar\phi^\mathbf{y})=\sum_{n,m=0}^\infty n\mathcal{W}_\Delta(\phi^{n-1}\bar{\phi}^m)^{\vec{z}_{n-1}\vec{y}_m}\psi^{(n,\bar{m})}_{x,\vec{z}_{n-1}\vec{y}_m} \\
\partial_{\bar\phi^x}\Psi(\phi^{\mathbf{x}},\bar\phi^\mathbf{y})=\sum_{n,m=0}^\infty m\mathcal{W}_\Delta(\phi^{n}\bar{\phi}^{m-1})^{\vec{z}_{n}\vec{y}_{m-1}}\psi^{(n,\bar{m})}_{\vec{z}_{n}\vec{y}_{m-1},x} 
\end{gather}
Here each derivative must be understood in the Wirtinger sense 

\begin{equation}
\label{eq:wirtinger}
\partial_{\phi^x}=\frac{1}{\sqrt{2}}\big(\partial_{\varphi^x}+i \partial_{\pi^x} \big),\hspace{0.5 cm} \partial_{\bar\phi^x}=\frac{1}{\sqrt{2}}\big(\partial_{\varphi^x}-i \partial_{\pi^x} \big)
\end{equation}
On the other hand, the associated Skorokhod integrals act on (anti)holomorphic functionals as: 

\begin{align}
\label{eq:holomorphicSkorokhod}
\partial^{*\phi^x} =-\Delta^{xy}\partial_{\bar\phi^y}+\phi^x, & &\partial^{*\phi^x}\mathcal{W}_\Delta(\phi^{n}\bar{\phi}^m)^{\vec{z}_{n}\vec{y}_m}= \mathcal{W}_\Delta(\phi^{n+1}\bar{\phi}^m)^{x,\vec{z}_{n}\vec{y}_m}\nonumber \\
\partial^{*\bar\phi^x} =-(\Delta^*)^{xy}\partial_{\phi^y}+\bar\phi^x, & &\partial^{*\bar\phi^x}\mathcal{W}_\Delta(\phi^{n}\bar{\phi}^m)^{\vec{z}_{n}\vec{y}_m}= \mathcal{W}_\Delta(\phi^{n}\bar{\phi}^{m+1})^{\vec{z}_{n}\vec{y}_m,x}
\end{align}

In the same  form as in the real case, we can write down the particle number operator $N$ and the $U(1)$ charge $Q$ encoding the matter/antimatter content of the states:
\begin{align}
\label{eq:ChargeAndNumber}
    N=\partial^{*\phi^x}\partial_{\phi^x}+ \partial^{*\bar\phi^x}\partial_{\bar\phi^x} & & Q=\partial^{*\phi^x}\partial_{\phi^x}- \partial^{*\bar\phi^x}\partial_{\bar\phi^x}
\end{align}

With those operators we see that the decomposition of \eqref{eq:complexchaosDecomposition} is such that

\begin{equation}
\label{eq:descompositionCharge}
L^2\big({\mathcal{N}_{\mathbb{C}}'},D\mu_c\big)=\sum_{n=0}^\infty\oplus \left( \sum_{m=0}^n\oplus \mathcal{H}_\Delta^{(n-m,\bar{m})} \right)
\end{equation}
where $\mathcal{H}_\Delta^{(n-m,\bar{m})}$ are simultaneous eigenspaces of $N$ and $Q$ with eigenvalues $n$ and $n-2m$ respectively. 

 Let us briefly introduce for later convenience Gross-Sobolev spaces \cite{huAnalysisGaussian2016} that extend the concept of \eqref{eq:MalliavinSovoleb} to any $L^p$-norm and $m$ number of derivatives. This represent the generalization to Gaussian analysis to the standard notion of Sobolev space, see \textit{e.g.} \cite{trevesTopologicalVector1967}. The $L^p$ norm $\lVert  \cdot \rVert_{\mu_c,p}$ of a function $\Psi(\phi^{\mathbf{x}},\bar\phi^{\mathbf{x}})$  is 
$$\lVert \Psi \rVert_{\mu_c,p} = \left(\int_{\mathcal{N}'_{\mathbb{C}}}D\mu(\phi^\mathbf{x})\lvert \Psi(\phi^{\mathbf{x}},\bar\phi^{\mathbf{x}}) \rvert^p\right)^{\frac1p}$$

 \begin{definition}{\textbf{(Gross-Sobolev spaces)}}
     \label{def:grossSobolev}
    The holomorphic Gross-Sobolev space $\mathbb{D}^{p,m}_{\mu_c}$ is the Banach space for integers $p\geq 1$ and $m\geq0$ is defined as 
        \begin{equation}
        \label{eq:GrossSobolev2}
        \mathbb{D}_{\mu_c,h}^{p,m}=\left\{\Psi(\phi^\mathbf{x}) \textrm{ s.t. } \lVert (1+N)^{\frac{m}2} \Psi \rVert_{\mu_c,p}< \infty \right\}
    \end{equation}
    For the special case of $m=0$ we denote $\mathbb{D}_{\mu_c}^{p,0}=L^p_{Hol}(\mathcal{N}'_{\mathbb{C}},D\mu_c)$, the Holomorphic $L^p$ Banach space of the Gaussian measure $\mu_c$. It can be shown \cite{huAnalysisGaussian2016} that $\mathbb{D}_{\mu_c}^{p,m}$ is the domain where $$\partial_{\phi^{\vec{\mathbf{x}}_m}}^m: \mathbb{D}_{\mu_c}^{p,m}\longrightarrow \mathcal{H}_\Delta^{\otimes m}\otimes L^p_{Hol}(\mathcal{N}'_\mathbb{C},D\mu_c)$$ can be defined as a bounded operator. See \autoref{lemm:Banach} and the comments below.  
 \end{definition}

 Gross-Sobolev spaces fulfill  $\mathbb{D}_{\mu_c}^{p,m+1}\subset \mathbb{D}_{\mu_c}^{p,m}$ and, due to Hölder's inequality for probability measures,  $\mathbb{D}_{\mu_c}^{p+1,m}\subset \mathbb{D}_{\mu_c}^{p,m}$  . In particular $L^q(\mathcal{N}'_{\mathbb{C}},D\mu_c)\subset L^p(\mathcal{N}'_{\mathbb{C}},D\mu_c)$ for $p \leq q$.

\subsection{Modelling pure quantum field  states II: Hida test functions and Hida calculus}
\label{sec:Hida}

    The set of Hida test functions   represents a \textit{second quantized} (s.q.)  notion of test functions that may be considered just an additional example of NF space. What we mean as second quantized construction is to consider that these Hida test functions are defined as  functions on a DNF space $\mathcal{N}'_{\mathbb{C}}$ with coefficients in its dual $\mathcal{N}_{\mathbb{C}}$. We denote $(\mathcal{N}_{\mathbb{C}})$ to the set of Hida test functions built upon $\mathcal{N}_{\mathbb{C}}$ and $(\mathcal{N}_{\mathbb{C}})'$ to its strong dual, the set of Hida distributions.\footnote{Notation in this case, even though standard, could be misleading. Notice  that $(\mathcal{N}_{\mathbb{C}})$ represents Hida test functions with domain $\mathcal{N}'_{\mathbb{C}}$.}

    Due to its relevance in the description of the space of pure quantum field states we have prepared a specific section for them and their tools. This concept of \textit{s.q.} test function is far from unique and we refer to \cite{hidaWhiteNoise1993} and \ref{app:Hida} for other definitions and a deeper study of their structure. In this work we use $(\mathcal{N}_{\mathbb{C}})$ because there is a straightforward generalization of Malliavin calculus to the space of Hida distributions $(\mathcal{N}_{\mathbb{C}})$'.

We will introduce $(\mathcal{N}_{\mathbb{C}})$ using rigged Hilbert spaces. Let us consider first the Hilbert space which we will use in the construction of a Gel'fand triple at the level of $\mathcal{N}_{\mathbb{C}}$. Using \autoref{def:FN} we know that the topology of the latter is provided by a sequence of Hilbert spaces $(\mathcal{H}_p, \langle \cdot, \cdot \rangle_p)$ and embeddings $\mathcal{H}_q\hookrightarrow \mathcal{H}_p$, which are Hilbert-Schmidt. We will use these ingredients to build another sequence of spaces and Hilbertian norms at the level of the functions on $\mathcal{N}'_{\mathbb{C}}$ and define, via a projective limit, the corresponding nuclear space structure for $(\mathcal{N}_{\mathbb{C}})$ see \eqref{eq:Hidatestfunction}.

 Let us consider the Gel'fand triple associated to the case when $\mathcal{N}_{\mathbb{C}}$ are complex test functions of compact support ({\it i.e.} the complex counterpart of equation \eqref{eq:riggedOnAManifold}) with respect to the Hilbert space $\mathcal{H}_{Vol}=L^2(\Sigma,\mathrm{dVol}_h)$. This Hilbert space is  obtained as the completion of $C^\infty_c(\Sigma)$ with respect to the scalar product induced by the measure $\mathrm{dVol}_h$ associated with a Riemannian metric $h$ on the hypersurface $\Sigma$. Thus we are considering the triple 
\begin{equation}
    \label{eq:riggedOnAManifoldcomp}
    \mathcal{N}_{\mathbb{C}}\subset L^2(\Sigma,\mathrm{dVol}_h) \subset \mathcal{N}_{\mathbb{C}}^{\prime*}
    \end{equation}

By using this rigged space, we can represent the space of test functions $C^\infty_c(\Sigma)_{\mathbb{C}}$ as a suitable linear conjugate subspace of the space of distributions $\mathcal{N}_{\mathbb{C}}\subset D'(\Sigma)_{\mathbb{C}}^*$. We shall use later this set as the set of classical fields.  Thus, our pure quantum states will be written then as test functions on this manifold $\mathcal{N}'_{\mathbb{C}}=D'(\Sigma)_{\mathbb{C}}$. Our notion of test function is provided by a nuclear topology and the injection of that topology into a Gel'fand triple respecting continuity.  This is the construction used in White Noise analysis \cite{hidaWhiteNoise1993,holdenStochasticPartial2010,westerkampRecentResults2003}. The topology of the space of test functions does not really depend on the particular Gel'fand triple that we use to realize the duality.  Because of this, it is common to introduce them as a subset of the white noise space of square integrable functions $L^2(\mathcal{N}'_{\mathbb{C}},D\beta)$ where $\beta$ is the white noise measure whose characteristic functional is

\begin{equation}
\label{eq:whiteNoise}
\int_{\mathcal{N}'_{\mathbb{C}}}D\beta(\phi^\mathbf{x}) e^{i(\bar\rho_x\phi^x+\rho_x\bar\phi^x)}= e^{-\bar\rho_x\delta^{xy}\rho_y}=\exp\left(-\int_\Sigma \textrm{dVol}(x)\bar\rho(x)\rho(x)\right)
\end{equation}
Moreover, from this definition we have that $\mathcal{H}_{Vol}$ becomes the Cameron-Martin Hilbert space of $L^2(\mathcal{N}'_{\mathbb{C}},D\beta)$.

We will treat holomorphic and antiholomorphic functions separately. Thus we start by considering  holomorphic polynomials with coefficients in  $\mathcal{N}_{\mathbb{C}}$ as the minimal space of test functions\footnote{For the case of Homogeneous polynomials we can even endow this space with a nuclear topology \cite{dineenComplexAnalysis1981}.} in $L^2_{Hol}(\mathcal{N}'_{\mathbb{C}},D\beta)$. Recall that an holomorphic polynomial of degree $n$ with coefficients in $\mathcal{N}_{\mathbb{C}}$ can be written as
\begin{equation}
\label{eq:polynomials} 
P_n(\phi^{\mathbf{x}})= C^0+C^1_x\phi^x+C^{2}_{x_1x_2}(\phi^2)^{x_1x_2}+\cdots+C^n_{\vec{x}_n}(\phi^n)^{\vec{x_n}}
\end{equation}    
Consequently, we must consider  the algebra of holomorphic polynomials $\mathcal{P}_{Hol}(\mathcal{N}'_{\mathbb{C}})$, which we know that  form a dense subset of the Hilbert space $L^2_{Hol}(\mathcal{N}', D\beta)$. The next step is to enlarge the set $\mathcal{P}_{Hol}(\mathcal{N}'_{\mathbb{C}})$ to obtain a more operative notion of test function.   To that end we make use of the topology of $\mathcal{N}_{\mathbb{C}}$ to build a sequence of Hilbert spaces and injections, and define our final set as the projective limit of that sequence, as it was explained in \autoref{def:FN}.

 Remember that $\mathcal{N}_{\mathbb{C}}$ is a nuclear space that we see as a subset of $\mathcal{H}_{Vol}$, and therefore it admits the sequence of Hilbert spaces, scalar products and Hilbert-Schmidth injections described in \autoref{def:FN} such that $\mathcal{H}_0=\mathcal{H}_{Vol}$. If we denote $\mathcal{H}_{p,q}$ with  $0\leq p,q$ to the weighted Hilbert spaces with norm $\| \cdot \|_{p,q}={2}^\frac{q}{2}\| \cdot \|_p $ and $\| \cdot \|_0=\| \cdot \|_{Vol}$. These norms are equivalent for the Hilbert space $\mathcal{H}_{p,q}$ but they may affect the convergence in the $q$-parametric sequence of Bosonic Fock spaces $\Gamma \mathcal{H}_{p,q}$. The space of Hida test functions is built using the following projective limits 
\begin{equation}
    \label{eqn:limitsHida}
  \begin{array}{rcccccccl}
      &\Gamma \mathcal{H}_{\infty,0}& \hookrightarrow  &\cdots & \hookrightarrow & \Gamma \mathcal{H}_{p,0}  & \hookrightarrow  &
      \Gamma \mathcal{H}_{Vol}
      &\cong L^2_{Hol}(\mathcal{N}'_{\mathbb{C}},D\beta) \\
      &\rotatebox[origin=c]{90}{$\hookrightarrow$} & & & & \rotatebox[origin=c]{90}{$\hookrightarrow$} 
      & &\rotatebox[origin=c]{90}{$\hookrightarrow$} \\
      &
      \Gamma \mathcal{H}_{\infty,q} & \hookrightarrow  &\cdots & \hookrightarrow & \Gamma \mathcal{H}_{p,q} & \hookrightarrow &\Gamma \mathcal{H}_{0,q} 
      \\
      &\rotatebox[origin=c]{90}{$\hookrightarrow$} & & & & \rotatebox[origin=c]{90}{$\hookrightarrow$} 
      & &\rotatebox[origin=c]{90}{$\hookrightarrow$} \\
      &\vdots & & & & \vdots
      & & \vdots \\
      &\rotatebox[origin=c]{90}{$\hookrightarrow$} & & & & \rotatebox[origin=c]{90}{$\hookrightarrow$} 
      & &\rotatebox[origin=c]{90}{$\hookrightarrow$} \\
      (\mathcal{N}_{\mathbb{C}})\cong 
      &
      \Gamma \mathcal{H}_{\infty,\infty} & \hookrightarrow  &\cdots & \hookrightarrow & \Gamma \mathcal{H}_{p,\infty} & \hookrightarrow &\Gamma \mathcal{H}_{0,\infty} 
  \end{array}
\end{equation}
We refer to \cite{kondratievGeneralizedFunctionals1996,obataWhiteNoise1994, hidaWhiteNoise1993, kuoWhiteNoise1996,westerkampRecentResults2003} for detailed studies of these limits. The space of Hida test functions $(\mathcal{N}_{\mathbb{C}})$ is then obtained as the projective limit of that sequence, and then we can write
\begin{equation}
\label{eq:Hidatestfunction}
(\mathcal{N}_{\mathbb{C}})=\bigcap_{p,q=0}^\infty  \mathcal{I}^{-1}\Big(\Gamma \mathcal{H}_{p,q}\Big).
\end{equation}
Where $\mathcal{I}$ is the Seagal isomorphism \eqref{eq:Segal} with respect to the White noise measure $\beta$. It is important to remark that the topology of $(\mathcal{N}_{\mathbb{C}})$ depends only on the topology of the nuclear space $\mathcal{N}_{\mathbb{C}}$ and not on the particularities of the white noise measure $\beta$ (see \cite{kondratievGeneralizedFunctionals1996, westerkampRecentResults2003}). This is a consequence of the independence of the chaos decomposition  on the Gaussian measure  for holomorphic functions that we noted in \eqref{eq:holchaosdecomposition}.

With these spaces in hand, we can define the  White Noise  Gel'fand triple introducing Hida distributions 
\begin{equation}
\label{eq:Triple_hida}
(\mathcal{N}_{\mathbb{C}})\subset L^2_{Hol}(\mathcal{N}'_{\mathbb{C}}, D\beta)\subset (\mathcal{N}_{\mathbb{C}})^{\prime *},
\end{equation}

     We can describe the dual space $(\mathcal{N}_{\mathbb{C}})'$, in an analogous way to the definition of $(\mathcal{N}_{\mathbb{C}})$, as the injective limit
    \begin{equation}
      \label{eq:Hidatestfunctionduak}
      (\mathcal{N}_{\mathbb{C}})'=\bigcup_{p=0}^\infty \mathcal{I}^{-1}\Big(\Gamma \mathcal{H}_{-p,-q}^*\Big).
    \end{equation}   
   Here the Banach space  $\mathcal{H}_{-p,-q}$ is identified with the dual of $\mathcal{H}_{p,q}$ in the Ge'lfand triple $
    \mathcal{H}_{p,q} \subset \mathcal{H}_{Vol} \subset \mathcal{H}_{-p,-q}^*
    $ and $\mathcal{I}^{-1}$ is extended to provide a chaos decomposition in this distributional context using Hida Malliavin-calculus techniques that we present below. Again we refer to  \cite{ hidaWhiteNoise1993, kuoWhiteNoise1996, obataWhiteNoise1994, kondratievGeneralizedFunctionals1996} for further details.

  Notice that as $(\mathcal{N}_{\mathbb{C}})$ is a nuclear space, it is a convenient space to define a differential calculus on a manifold with it as a model and,  using it, tensor fields and Hamiltonian structures.

  As anticipated, the role of the white noise measure $\beta$ is auxiliary in the construction. If we consider any other Gaussian measure $\mu_c$ with covariance $\Delta^{\mathbf{xy}}$ such that there exists a $q\geq 0$ and a norm of the chain $\| \cdot \|_{p,q}$ such that $\lvert\bar\rho_x\Delta^{xy}\rho_y \rvert \leq  \| \rho_{\mathbf{x}} \|_{p,q}^2 $, then according to the characterization theorems \cite{kondratievGeneralizedFunctionals1996} we also have 
  \begin{equation}
    \label{eq:Triple_hida_second}
    (\mathcal{N}_{\mathbb{C}})\subset L^2_{Hol}(\mathcal{N}'_{\mathbb{C}}, D\mu_c)\subset (\mathcal{N}_{\mathbb{C}})^{\prime *},
    \end{equation}
An equivalent derivation can be considered for any Gaussian measure, as for instance the one corresponding to the vacuum of the quantum theory.
 We call this case \textbf{the second quantization Gel'fand triple}. We must also consider the decomposition for the antiholomorphic part. Thus, we consider the triple  
 \begin{equation}
     \label{eq:gelfandtriples}
      (\mathcal{N}_\mathbb{C})\otimes (\mathcal{N}_\mathbb{C})^* \subset L^2(\mathcal{N}'_{\mathbb{C}}, D\mu_c)\subset (\mathcal{N}_{\mathbb{C}})^{\prime \ast}\otimes (\mathcal{N}_\mathbb{C})'.
     \end{equation}
      {Hence,} Hida test functions respect the decomposition in holomorphic and antiholomorphic parts of \eqref{eq:ComplexPolinomials}.

 We list below some important properties that and  showcase the power of the Hida triple. For different presentations and proofs we refer to  \cite{kondratievGeneralizedFunctionals1996,hidaWhiteNoise1993,kuoWhiteNoise1996,obataWhiteNoise1994,westerkampRecentResults2003}.
 
 \begin{itemize}
   \item Let $\Psi\in (\mathcal{N}_\mathbb{C})\otimes (\mathcal{N}_\mathbb{C})^*$ the coefficients of its chaos decomposition fulfill $\psi^{(n,\bar{m})}_{\vec{\mathbf{x}},\vec{\mathbf{y}}}\in \mathcal{N}_\mathbb{C}^{\hat\otimes n}\otimes \mathcal{N}_\mathbb{C}^{* \hat\otimes m}$. Moreover, $\Psi(\phi^\mathbf{x})$ is a pointwise defined continuous function. 
     \item $(\mathcal{N}_{\mathbb{C}})$ and $(\mathcal{N}_\mathbb{C})\otimes (\mathcal{N}_\mathbb{C})^*$ are algebras under pointwise multiplication. Moreover, the algebras $\mathcal{P}$, $\mathcal{T}$ and $\mathcal{C}$ introduced in Section \ref{sec:algebras} are subalgebras of one or both spaces. See \ref{app:Hida} for a proof.
   \item   Hida test functions are infinitely differentiable and they and their derivatives belong to $L^p_{Hol}(\mathcal{N}'_{\mathbb{C}},D\mu_c)$ $\forall\, p\geq 1$. This is encoded in $$(\mathcal{N}_\mathbb{C})\subsetneq \bigcap\limits_{p\geq1, m\geq 0}\mathbb{D}^{p,m}_{\mu_c}:= \mathbb{D}_{\mu_c}$$
    See \ref{app:Hida} for a proof. \footnote{$\mathbb{D}_{\mu_c}$ is often called the Meyer-Watanabe space of test functions. We refer to \cite{huAnalysisGaussian2016,hidaWhiteNoise1993} for further details.} 
   \item  As a consequence of the previous property and the fact that that the dual of a $L^p$-space is a $L^q$-space with $\frac{1}p+\frac{1}q= 1$ we get 
   $$L^p(\mathcal{N}'_{\mathbb{C}},D\mu_c)\subset (\mathcal{N}_{\mathbb{C}})^{\prime \ast}\otimes (\mathcal{N}_\mathbb{C})^\prime\qquad \forall\, p\geq 1$$

 \end{itemize}

 Below we list some examples of Hida test functions and distributions that allow to operate in more general contexts in Gaussian analysis and will be of great importance in \cite{alonsoGeometricFlavours2024}.

\begin{example}{\bf (Hida-Malliavin calculus)}
    \label{ex:HidaMalliavin}
    The generator of the $S_\mu$ transform is a Hida test function
    $$\exp{\big(\overline{{\rho}_x\phi^x}+{\rho_x{\phi}^x}-\bar\rho_x\Delta^{xy}\rho_y\big)}\in (\mathcal{N}_\mathbb{C})\otimes (\mathcal{N}_\mathbb{C})^*.$$  As a consequence $S_\mu[\Psi]$ can be extended  to any Hida distribution $\Psi\in(\mathcal{N}_{\mathbb{C}})^{\prime \ast}\otimes (\mathcal{N}_\mathbb{C})'$. In particular this means that we can find a chaos decomposition for any Hida distribution with a careful use of \eqref{eq:wienerItoSeagalwithTtransform}.  
    This turns Hida distributions into easier-to-handle objects that other notions of \textit{s.q} distributions. In particular, this feature allows for an efficient characterization of Hida test functions and distributions using the celebrated characterization theorems \cite{kondratievGeneralizedFunctionals1996,hidaWhiteNoise1993} and \ref{app:Hida}.  

    The generalization of the machinery of chaos decompositions to general distributions allows us to extend most of the tools of Malliavin calculus, described in \autoref{sssec:MaliiavinCalculus}, to distributions in $(\mathcal{N}_{\mathbb{C}})^{\prime \ast}\otimes (\mathcal{N}_\mathbb{C})'$. This type of calculus is called Hida-Malliavin calculus and we refer elsewhere for a thorough construction \cite{nunnoMalliavinCalculus2009,holdenStochasticPartial2010}.
\end{example}

A common problem that is solved by Hida calculus is to make sense out of Radon-Nikodym derivatives. This will be particularly relevant when QFT theories with respect to evolving backgrounds are considered since these properties of Hida calculus will provide us with tools to compare the states of the quantum fields with respect to them. With these applications in mind, we briefly comment now a few useful properties: 

\begin{example}{\bf (Translation of a Gaussian measure)}
    \label{ex:translations} A general result of Gaussian analysis, known as Cameron-Martin theorem, states that,  a translation in the domain by a constant element $\chi^\mathbf{x}\in \mathcal{N}'_\mathbb{C}$ provides a Gaussian measure mutually singular with the original  one, unless $\chi^\mathbf{x}\in \mathcal{H}_K$, the Cameron-Martin Hilbert space. This implies that the \textit{would be Radon-Nikodym derivative}
    $$\frac{D\mu_c(\phi^\mathbf{x}-\chi^\mathbf{x})}{D\mu_c(\phi^\mathbf{x})}=\exp{\big(\bar\phi^xK_{xy}\chi^y+\bar\chi^xK_{xy}\phi^y-\bar\chi^xK_{xy}\chi^y\big)}$$
does not exist as a $L^1(\mathcal{N}'_\mathbb{C},D\mu_c)$ function. The power of Hida calculus is to extend its meaning to general $\chi^\mathbf{x}\in \mathcal{N}'_\mathbb{C}$ as an element of $(\mathcal{N}_{\mathbb{C}})^{\prime *}\otimes (\mathcal{N}_\mathbb{C})'$ and provide a tool to perform  calculus with the generalized Wiener-Ito decomposition. 
\end{example}

\begin{example}{\bf (Mean zero Gaussian measure)}
\label{ex:meanZeroGaussian} In general, Gaussian measures with different covariances, even if they both fulfill an estimate of the type $\lvert\rho_x\Delta^{xy}\rho_y \rvert \leq  \| \rho_{\mathbf{x}} \|_{p,q}^2 $, are mutually singular.  

 When working with the space of Hida test functions, if $\nu_c$ is
 another mean zero Gaussian measure with covariance $\Omega^{\mathbf{xy}}$ then 
$$S_{\mu_c}\Big[\frac{D\nu_c}{D\mu_c}\Big](\rho_\mathbf{x})=e^{\bar\rho_x(\Omega^{xy}-\Delta^{xy})\rho_y},$$
where $\frac{D\nu_c}{D\mu_c}\in (\mathcal{N}_{\mathbb{C}})^{\prime *}\otimes (\mathcal{N}_\mathbb{C}^*)$ makes sense as a Hida distribution even if it does not  as a function.  
\end{example}

\begin{example}{\bf (Reproducing Kernel)}
    \label{ex:repKernel} Coherent states of the holomorphic subspace $\mathcal{K}_{\chi}(\phi^{\mathbf{y}})\in L^2_{Hol}\big({\mathcal{N}_{\mathbb{C}}'},D\mu_c\big)$ can be extended to elements $\rho^{\mathbf{x}}\in \mathcal{H}_\Delta$ as $ \mathcal{K}_\rho(\phi^{\mathbf{y}})=e^{\rho^xK_{xy}\phi^x}$ that makes sense as an element of  $L^2_{Hol}\big({\mathcal{N}_{\mathbb{C}}'},D\mu_c\big)$ but, moreover 
    $$S_{\mu_c}\Big[ e^{\bar\sigma^xK_{xy}\phi^x}\Big](\rho_\mathbf{x})=e^{\overline{\sigma^x\rho_x}}$$
    Then  $e^{\bar\sigma^xK_{xy}\phi^x}\in (\mathcal{N}_{\mathbb{C}})'$ 
    is a Hida distribution and works as a reproducing kernel when it is paired to $(\mathcal{N}_{\mathbb{C}})\subset L^2_{Hol}\big({\mathcal{N}_{\mathbb{C}}'},D\mu_c\big)$. Thus  for $\Psi\in  (\mathcal{N}_{\mathbb{C}})$
    \begin{equation}
    \label{eq:rpekernelHida}
    \int_{\mathcal{N}_{\mathbb{C}}'}D\mu_c(\phi^{\mathbf{x}})\overline{e^{\bar\sigma^xK_{xy}\phi^x}}\Psi(\phi^{\mathbf{x}})= \Psi(\sigma^\mathbf{x})
    \end{equation}
    \end{example}

    A question that deserves more attention is why we choose white noise instead of  \textbf{the second quantization Gel'fand triple} if the latter posses richer physical meaning when we interpret $\mu_c$ as part of the vacuum of our quantum theory.   We could as well present Hida test functions and distributions using this triple, but the fact that the topology of $(\mathcal{N}_{\mathbb{C}})$ only depends on that of $\mathcal{N}_{\mathbb{C}}$ turns  white noise into an auxiliary but simplifying way of presenting it. Moreover, the vacuum of a theory is a particular feature of such theory and therefore we cannot assume prior knowledge of the measure $\mu_c$ that is part of the vacuum. In this way, in order to consider the most general setting, it is better to model pure states as the common dense subspace of Hida test functions $(\mathcal{N}_{\mathbb{C}})$. In that case we can write down the equations of motion and the operators of the theory only referred to this set of test functions, as we will consider in \cite{alonsoGeometricFlavours2024}.  However, these equations cannot be solved in  $(\mathcal{N}_{\mathbb{C}})$ and the operators considered to describe observables will have naturally larger domains. In order to recover the full space, we will use our \textit{ a posteriori} knowledge of the  Gaussian measure that endows $(\mathcal{N}_{\mathbb{C}})$ with a pre-Hilbert topology and complete it to recover the whole space.  This procedure is similar to the protocol to recover the one particle Hilbert space from a classical field theory used in \cite{waldQuantumField1994}. In that case we can model the classical theory over a space of test functions $\mathcal{N}$ and then complete it  with a Hilbertian norm that depends on the particular theory we are considering. Notice though that the treatment in \cite{waldQuantumField1994} is covariant and here we aim to model a non covariant approach. This is important to notice in regard of the treatment of the equations of motion. Our discussion on how to introduce evolution is radically different from covariant approaches and is the subject of study of  \cite{alonsoGeometricFlavours2024}.

\section{Adapting Classical Field Theory to quantization}
\label{sec:CFTQuantizar}

In this section we will adapt  the geometrical description of a classical field theory to quantization, although the quantization procedure itself will be dealt with in the second part \cite{alonsoGeometricFlavours2024} of this work.
 
 Geometrical descriptions of classical field theories in terms of jet bundles also exist and admit  well defined Hamiltonian formalisms. Nonetheless, our goal is to build, from this classical manifold, a geometrical description of the space of quantum states on the space of fields, and our choice of a field phase-space provides us with a simpler analog of the finite dimensional construction for quantum mechanical systems (see \cite{kibbleGeometrizationQuantum1979}).

  For the spacetime structure in this and the subsequent paper we consider a $d+1$ globally hyperbolic spacetime $(\mathcal{M},\mathbf{g}, \nabla_g)$
 endowed with a pseudo-Riemannian structure $\mathbf{g}$ of Lorentzian 
signature, with sign convention  $(-,+,\cdots,+)$, and the Levi-Civita connection $\nabla_g$. Because of its condition of globally hyperbolic spacetime $\mathcal{M}$ is diffeomorphic to $\mathbb{R}\times\Sigma$ where $\Sigma$ is a $d-$dimensional manifold diffeomorphic
to every space-like Cauchy hypersurface of $\mathcal{M}$ and $\mathbb{R}$ represents the time parameter $t$. We will develop further this description in \cite{alonsoGeometricFlavours2024}. For our purposes in this work, we take $\Sigma$ compact and denote ${h}_t$ to the $t-$parametric 
 Riemannian metric induced from $\mathbf{g}$  by projection on $\Sigma$.

We will focus on the real scalar field, i.e. the phase space of the theory will be given by the cotangent bundle to a real line bundle $\pi_{\Sigma L}:L\to \Sigma$. For simplicity we will consider the trivial case in which $L=\mathcal{N}$ is a space of smooth functions over $\Sigma$.

Furthermore, the set of fields modelled by any of these function spaces should be solutions of a certain dynamical field equation representing the type of field. The simplest example will be the case of a Klein-Gordon field theory. In this work we will consider only linear equations and therefore the corresponding manifold of fields will admit a linear structure which will significantly simplify the geometrical setting. 

\subsection{Choosing the space of classical fields and their momenta}

If $\mathcal{N}$ represents the set of fields, the corresponding phase space should correspond to the cotangent bundle $T^*\mathcal{N}\simeq \mathcal{N}\times \mathcal{N}'$, where we used the linearity of $\mathcal{N}$ and the corresponding triviality of the cotangent bundle. The coordinate representation of an element of this space is given by $(\varphi_\mathbf{x},\pi^\mathbf{y})$, with the coordinate systems described above. This space is endowed with a weakly symplectic form  which takes the local expression:

\begin{equation}
\label{eq:symplecticForm}
\omega_{T^*\mathcal{N}}=  d\pi^x \wedge d\varphi_x.
\end{equation}

This situation is problematic in the sense that the model space for the fiber and the base of the cotangent bundle $T^*\Gamma(\Sigma,L)$ are different. To quantize the system is convenient to express both coordinates over the same model space. To do so we endow $\Sigma$ with a Riemannian structure $h$, use it to define a measure $d\mathrm{Vol}$ and use the natural structure of rigged Hilbert space given by 
\begin{equation}
\label{eq:rhsOnSigma}
\mathcal{N}\subset L^2(\Sigma,d\textrm{Vol})\subset \mathcal{N}'
\end{equation} 

As a result, the elements of $\mathcal{N}$ corresponds to a linear subspace of its dual and we can write the coordinate representation of the image of an element of $\mathcal{N}$ as $\delta^{xy}\varphi_x$. Thus, in this setting, the manifold of fields must  be $\mathcal{M}_F=\mathcal{N}'\times \mathcal{N}'$ and coordinates $(\varphi^\mathbf{x},\pi^\mathbf{y})$ are to be understood as Darboux coordinates for a densely defined weakly symplectic structure
\begin{equation}
\label{eq:symplecticFormlifted}
\omega_{\mathcal{M}_F}=\delta_{xy} d\pi^x \wedge d\varphi^y =\int_\Sigma \frac{d^dx}{\sqrt{\lvert h\rvert}}\ [d\pi(x)\otimes d\varphi(x)-d\varphi(x)\otimes d\pi(x)],
\end{equation}
where now $\pi(x),\varphi(x)$ are densities of weight 1. Notice that this symplectic structure is defined over the dense subspace $L^2(\Sigma,d\textrm{Vol})\subset \mathcal{N}'$. If we consider the space of polynomial functions  $\mathcal{P}(\varphi^\mathbf{x},\pi^\mathbf{y})$ with coefficients in $\mathcal{N}$ ( that, eventually, can be  extended to more general completions) we can also consider a Poisson bivector of the form:
\begin{multline}
\label{eq:poissonbracket}
\left\{P,Q\right\}_{\mathcal{M}_F}=
\delta^{xy}\left( \partial_{\varphi^x}P\partial_{\pi^y}Q-\partial_{\pi^x}P\partial_{\varphi^y}Q \right):=\\
\int_\Sigma d\textrm{Vol}(x)\left( \frac{\partial P}{\partial\varphi(x)}\frac{\partial Q}{\partial\pi(x)}-\frac{\partial P}{\partial\pi(x)}\frac{\partial Q}{\partial\varphi(x)} \right),
\end{multline}
where partial derivatives are understood as Fréchet derivatives, which are well defined for $\mathcal{P}(\varphi^\mathbf{x},\pi^\mathbf{y})$. This Poisson structure is often mistaken with a symplectic structure in QFT quantization schemes since, for linear spaces, the difference is subtle (see \cite{brunettiAdvancesAlgebraic2015} chapter 5 for further discussion on that matter). With this Poisson bracket at hand we can proceed as usual and write down the Hamiltonian dynamics of classical observables.  Classical observables are elements of a suitable completion of the set of polynomials in the field and momenta variables $\mathcal{P}(\varphi^\mathbf{x},\pi^\mathbf{y})$ and existence and uniqueness of such dynamics is, in general, strongly dependent on this completion.   Such problems, though, will not be covered in this work.

Summarizing: our classical phase space shall be considered to be a linear subspace of $\mathcal{M}_F=\mathcal{N}'\times \mathcal{N}'$. This subspace  corresponds to the injection of $\mathcal{N}$ by the rigged Hilbert space structure introduced on $L^2(\Sigma, d\textrm{Vol})$. This structure is considered to be canonical as $d\textrm{Vol}$ is the measure induced by the Riemannian structure of the Cauchy hypersurface $\Sigma$. On this space we can consider two (almost) equivalent structures: a (weakly) symplectic structure $\omega_{\mathcal{M}_F}$ and a Poisson tensor $\{ \cdot, \cdot \}_{\mathcal{M}_F}$.
It is important to remark a subtle aspect of the construction: while the symplectic form on $T^*\mathcal{N}$ (Equation \eqref{eq:symplecticForm}) is canonical, the one on $\mathcal{M}_C$ (Equation \eqref{eq:symplecticFormlifted}   is not, since it depends on the choice of measure $d\mathrm{Vol}$. The same is true for the Poisson bracket \eqref{eq:poissonbracket}. This aspect will be important when considering quantum classical hybrid models.

Now, we will learn to introduce another structure on it to define more conveniently the quantization of the classical model. In order to do that, we need to introduce a complex structure on $\mathcal{M}_F$.

\subsection{Introducing a complex structure}
In order to obtain a system suited for geometric quantization we need an additional (almost) complex structure $J_{\mathcal{M}_F}$ satisfying:
\begin{itemize}
	\item it is densely defined,
	\item is such that $J_{\mathcal{M}_F}^2=-\mathds{1}$ in its domain $D(J_{\mathcal{M}_F})$,
 \item   it leaves $D(J_{\mathcal{M}_F})\subseteq L^2(\Sigma,d\textrm{Vol})\subset \mathcal{N}'$ invariant,
 \item  and $\omega_{\mathcal{M}_F}( \cdot ,J_{\mathcal{M}_F}\cdot )=-\omega_{\mathcal{M}_F}( J_{\mathcal{M}_F}\cdot ,\cdot ).$ This condition can be written as $J_{\mathcal{M}_F}^\dagger= -J_{\mathcal{M}_F}$ and means that the complex structure is compatible with the symplectic form $\omega_{\mathcal{M}_F}$. This is the imaginary part of the hermitian product of $L^2(\Sigma,d\textrm{Vol})$. We will use this fact bellow to construct a different hermitian product describing our theory. 
\end{itemize}

We shall see  in \cite{alonsoGeometricFlavours2024} that the complex structure required for quantization will be determined by the classical Hamiltonian on the space of fields. Hence, in order to consider different possible classical dynamics, we must explore the space of possible complex structures on $\mathcal{M}_F$.  It is simple to verify that there are infinitely many different structures satisfying those conditions. To express the most general complex structure let $(\varphi^{\mathbf{x}},\pi^{\mathbf{y}})$ be Darboux coordinates for $\mathcal{M}_F$ and let us consider the matrix elements $(A^t)^{\mathbf{x}}_\mathbf{y}:= \delta_{u \mathbf{y}} A^u_v\delta^{v\mathbf{x}}$. Notice that the definition of these variables already depends on the choice of the Hilbert space structure for the Gel'fand triple.

An arbitrary (almost) complex structure on $\mathcal{M}_F$ is locally expressed in the canonical coordinates of the scalar field as 
\begin{multline}
\label{eq:complexStructure}
-J_{\mathcal{M}_F} =
(\partial_{\varphi^y},\partial_{\pi^y})\left( \begin{array}{cc}
A^y_x & \Delta^y_x \\
D^y_x & -(A^t)^y_x
\end{array} \right)  \left( \begin{array}{c}
d\varphi^{x} \\
d\pi^{x}
\end{array} \right)=\\
=  \partial_{\varphi^y} \otimes[A^y_xd\varphi^x+\Delta^y_xd\pi^x]+ \partial_{\pi^x}\otimes [D^y_xd\varphi^x-(A^t)^y_xd\pi^y]
\end{multline}
With  the multiplication convention $(A\Delta)^{\mathbf{x}}_{\mathbf{y}}=A^{\mathbf{x}}_z\Delta^z_{\mathbf{y}}$ the condition $J^2=-\mathds{1}$ is translated into  $A^2+\Delta D=-\mathds{1},\ \Delta^t=\Delta, D^t=D, A\Delta=\Delta A^t$ and $A^tD=DA$. Notice that $D$ is fixed once $\Delta$  and $A$ are known. Let $K^{\mathbf{x}}_{\mathbf{y}}$ be the inverse of $\Delta_\mathbf{y}^{\mathbf{x}}$, i.e. $\Delta^{\mathbf{x}}_zK^z_\mathbf{y}=\delta_\mathbf{y}^{\mathbf{x}}$ then $D=(iA^t+\mathds{1})K(iA-\mathds{1})$.

This complex structure, together with the symplectic form \eqref{eq:symplecticFormlifted} induces a (pseudo)-Riemannian structure $\mu_{J,\mathcal{M}_F}(\cdot,\cdot)=\omega_{\mathcal{M}_F}(\cdot,-J\cdot)$, which, in the coordinates above, reads: 
\begin{multline}
\label{eq:riemannianInduced}
\mu_{\mathcal{M}_F}=(d\varphi^{x},d\pi^{x}) \left( \begin{array}{cc}
    -D_{xy} & A^t_{xy} \\
    A_{xy} & \Delta_{xy}
    \end{array} \right) \left( \begin{array}{c}
        d\varphi^{y} \\
        d\pi^{y}
    \end{array} \right)=\\
    -D_{xy}d\varphi^x\otimes d\varphi^y+\Delta_{xy}d\pi^x\otimes d\pi^y+ A_{xy}(d\varphi^y\otimes d\pi^x+d\pi^x\otimes d\varphi^y)
\end{multline}
where $A_{{\mathbf{x}}{\mathbf{y}}}=\delta_{{\mathbf{x}}z}A^z_{\mathbf{y}}$. To obtain a Riemannian structure form here we must impose also $\Delta_{\mathbf{xy}}>0>D_{\mathbf{xy}}$. 

With a change of coordinates, locally given by
\begin{align}
\label{eq:changeOfCoordinates}
d{\tilde{\varphi}^x}&=d{\varphi^x} &
d{\tilde{\pi}^x}&=A^x_yd{\varphi^y}+\Delta^x_yd{\pi^y} 
\\
\partial_{\tilde{\varphi}^x}&=\partial_{\varphi^x}- (K A)^y_x\partial_{\pi^y} &
\partial_{\tilde{\pi}^x}&=K^y_x\partial_{{\pi}^y}
\end{align}
 we can put in a canonical way $J_{\mathcal{M}_F}=\partial_{\tilde{\pi}^x}\otimes d\tilde{\varphi}^x-\partial_{\tilde{\varphi}^x}\otimes d\tilde{\pi}^x$, $\mu_{\mathcal{M}_F}=K_{xy}(d\tilde{\varphi}^x\otimes d\tilde{\varphi}^y+d\tilde{\pi}^x\otimes d\tilde{\pi}^y)$ and $\omega_{\mathcal{M}_F}=K_{xy} d\tilde{\pi}^x\wedge d\tilde{\varphi}^y$. We  can cast the structures into this canonical form with other changes of variables, but we choose this one because it leaves $\varphi^\mathbf{x}$ (the manifold of fields) invariant. Notice that we are transforming the original cotangent bundle structure we began with as the new coordinates combine base and fiber coordinates of the original bundle structure.

Having chosen a complex structure, we can define an isomorphism for $\mathcal{M}_F$ to become a complex vector space $\mathcal{M}_C=\mathcal{M}_F^{\mathbb{C}}$. Notice that while $\mathcal{M}_F$ depends on the choice of the Hilbert space structure for the Gel'fand triple (and the original $T^*\mathcal{N}$), the definition of $\mathcal{M}_C$ depends also on the choice of the complex structure.  This fact will be important later. Notice, anyway, that the symplectic structures on $T^*\mathcal{N}'$, $\mathcal{M}_F$ and $\mathcal{M}_C$ are all diffeomorphic, but the diffeomorphisms depend on the metric structure $h$ on $\Sigma$ (for $\mathcal{M}_F$)  and the complex structure $J_{\mathcal{M}_C}$ (for $\mathcal{M}_C)$. 

 Having chosen the complex manifold $\mathcal{M}_C$, we consider holomorphic and antiholomorphic  coordinates which are locally written as $(\utilde\phi^{\mathbf{x}},\utilde{\bar\phi}^{\mathbf{y}})=\frac{1}{\sqrt{2}}(\tilde{\varphi}^{\mathbf{x}}-i\tilde{\pi}^{\mathbf{x}},\tilde{\varphi}^{\mathbf{y}}+i\tilde{\pi}^{\mathbf{y}})$. We use this sign convention for the holomorphic coordinates to ease the discussion on quantization of \cite{alonsoGeometricFlavours2024}.  In this new manifold the complex structure becomes  $J_{\mathcal{M}_C}$, which  is now written as: 

\begin{equation}
\label{eq:complexholom}
J_{\mathcal{M}_C}=-i \partial_{\utilde\phi^x}\otimes d\utilde\phi^x+i \partial_{\utilde{\bar{\phi}}^x}\otimes d\utilde{\bar{\phi}}^x.
\end{equation}
At the tangent space level we can always write the projection operators $P_{\pm}=\frac12(\mathds{1}\pm iJ)$ such that $P_{+}\partial_{\utilde\phi^{\mathbf{x}}}=\partial_{\utilde\phi^{\mathbf{x}}}$ and $P_{-}\partial_{\utilde{\bar\phi}^{\mathbf{x}}}=\partial_{\utilde{\bar\phi}^{\mathbf{x}}}$. We shall use these projections in \cite{alonsoGeometricFlavours2024}.

As the triple $(\mu_C, \omega_C, J_C)$ defines a Kähler structure on $\mathcal{M}_C$, there exists
a (non unique) globally defined on a dense subspace (because the manifold is linear) Kähler potential  $\mathcal{K}$

\begin{equation}
\label{eq:kahlerpotential}
\mathcal{K}(\utilde{\bar \phi}^x, \utilde\phi^x)=\utilde{\bar \phi}^x  K_{xy} \utilde\phi^y .
\end{equation}

From the potential we can define the corresponding  Kähler form:

\begin{equation}
\label{eq:kahlerform}
\omega_{\mathcal{M}_C}= i\partial \bar \partial \mathcal{K},
\end{equation}
where $\partial$ and $\bar \partial$ represent the Dolbeault operators.  With the Kähler form and the expression of $J_{\mathcal{M}_C}$ of \eqref{eq:complexholom} we obtain the expression of the Riemannian form $\mu_{\mathcal{M}_C}$ and the hermitian form: $$h_{\mathcal{M}_C}= \frac{\mu_{\mathcal{M}_C}-i\omega_{\mathcal{M}_C}}{2}= K_{xy}d\utilde{\phi}^x\otimes d\utilde{\bar{\phi}}^y,$$ 
with inverse $h^{-1}_{\mathcal{M}_C}=\Delta^{xy}\partial_{\utilde{\bar\phi}^x}\otimes\partial_{\utilde{\phi}^y}$.

Finally, the expression of the Poisson bracket \eqref{eq:poissonbracket} in holomorphic coordinates becomes:
\begin{equation}
\label{eq:poissonHol}
\left\{P,Q\right\}_{\mathcal{M}_C}=
i\Delta^{xy}\left( \partial_{\utilde\phi^x}P\partial_{\utilde{\bar\phi^y}}Q-\partial_{\utilde{\bar\phi}^x}P\partial_{\utilde{\phi}^y}Q \right).
\end{equation}

It is important to notice that the election of the rigged Hilbert space made in \eqref{eq:rhsOnSigma}, to express the system in Darboux coordinates, is different for holomorphic coordinates. In the first case, we built the triple from the real classical fields onto the real momentum fields. In holomorphic coordinates the natural triple is defined from the complexified  fields 
$\mathcal{N}_{\mathbb{C}}$ onto the complexified conjugate fields $\mathcal{N}_{\mathbb{C}}^{\prime \ast}$:

\begin{equation}
\label{eq:HolomorphicTriple}
\mathcal{N}_{\mathbb{C}}\subset \mathcal{H}_\Delta\cong \mathcal{H}_K^*\subset {\mathcal{N}_{\mathbb{C}}^{\prime\ast}}
\end{equation}
 where $\mathcal{H}_\Delta=\overline{(\mathcal{N}_{\mathbb{C}},\langle \cdot,\cdot\rangle_\Delta)}$ is the  Hilbert space obtained by completion of $\mathcal{N}_{\mathbb{C}}$ with the scalar product induced by  the tensor $\Delta$ as $\langle \chi_{\mathbf{x}},\xi_{\mathbf{y}}\rangle_\Delta=\bar{\chi}_x\Delta^{xy}\xi_y$.  This space is identified with the dense subset of $\mathcal{N}_{\mathbb{C}}^{\prime *}$ on which the hermitian form $h_{\mathcal{M}_C}$ is finite. This Hilbert space  has a physical interpretation for $\mathcal{H}_\Delta$ as the space corresponding to the \textit{one-particle-state} structure. The space comprises the initial conditions for all the solutions of 'positive frequency' that are to be interpreted as particles in contrast with the elements of the conjugate $\mathcal{H}_{\Delta}^*$ which represent the corresponding anti-particles.    

On the other hand, we can consider a similar construction from the dual perspective if we consider the completion of  the dense subset of $\mathcal{N}_{\mathbb{C}}^{\prime*}$ on which the scalar product $\langle \cdot,\cdot\rangle_K$ induced by the hermitian tensor $K$ is defined.  Both spaces $\mathcal{H}_\Delta$ and $\mathcal{H}_K^*$ are identified by means of the Riesz representation theorem,  realizing the duality  $(\mathcal{N}_{\mathbb{C}},\mathcal{N}_{\mathbb{C}}^{\prime *})$ with the scalar product  $\langle \cdot,\cdot\rangle_\Delta$.  In the following, we will consider them as the same space.

Furthermore, the elements of $\mathcal{H}_K^*$ will be important later when we build the Schrödinger and holomorphic pictures of our field theory. Indeed, when considering "wave-functions" on our manifold of classical fields, $\mathcal{H}_K^*$ is identified with the Cameron-Martin space represent the allowed translations of the unique Gaussian measure which has covariance $\Delta^{\mathbf{xy}}$. From that point of view, we will consider them to be the natural representation of the tangent space to our space of fields.

Summarizing: the classical phase space of fields is described as a Kähler complex manifold $\mathcal{M}_C$, which has associated a rigged Hilbert space structure encoding the physical \textit{one-particle-space} structure. In the second part of this series \cite{alonsoGeometricFlavours2024}, we will use these structures to define, by geometric quantization, a geometrical description of the quantum field theory corresponding to this  classical phase space of fields. We will also investigate, using the tools presented so far, different methods of quantization and the relations between them.

\section{Conclusions}
\label{sec:conclusions}

In this paper we have summarized and collected some important results of Gaussian analysis with a vision in quantum field theory. First, we have argued that the natural domain for integration are strong duals of Nuclear Frechèt spaces (DNF) denoted by $\mathcal{N}'$ and that the concept of rigged Hilbert space emerged from the necessity of identify the Nuclear Frechèt  (NF) spaces with subsets of their duals. 

 Then we defined a Gaussian measure $D\mu_c$ and we have studied Gaussian integration theory of functions with domain in the model DNF space. The space of functions analyzed is $L^2\big({\mathcal{N}_{\mathbb{C}}'}, D\mu_c\big)$ and we distinguished three subalgebras of functions $\mathcal{T}$, $\mathcal{P}$ and $\mathcal{C}$.  In the second part of this series \cite{alonsoGeometricFlavours2024} they will unveil different properties of the Hilbert spaces used to represent a QFT of the scalar field and its operators. 
 
 The first example of dense subalgebra are trigonometric exponentials $\mathcal{T}(\mathcal{N}'_\mathbb{C})$ that serve as a basis for an isomorphism with a reproducing kernel Hilbert space. This tool will be of huge importance in the second part of this series to establish the properties of the algebra of observables. Closely related is the algebra of coherent states $\mathcal{C}_{Hol}(\mathcal{N}'_\mathbb{C})$ that is dense in the subspace of square integrable holomorphic functions $L^2_{Hol}\big({\mathcal{N}_{\mathbb{C}}'}, D\mu_c\big)$ and helps to establish the analytical properties of that subspace as well as the decomposition  in holomorphic and antiholomorphic parts of the whole space.

 The last  dense subalgebra of $L^2\big({\mathcal{N}_{\mathbb{C}}'}, D\mu_c\big)$ are polynomials  $\mathcal{P}(\mathcal{N}'_\mathbb{C})$. This set helps to establish tools of huge importance for Gaussian analysis an its relation with quantum field theory. First, we established the Wiener-Ito decomposition theorem that provides an isomorphism with a bosonic Fock space completely determined by the properties of the Gaussian measure. This will yield the particle interpretation of the QFT in \cite{alonsoGeometricFlavours2024}. Secondly, we developed the Segal-Bargmann transform that relates $L^2_{Hol}\big({\mathcal{N}_{\mathbb{C}}'}, D\mu_c\big)$ with its counterpart with real domain $L^2\big({\mathcal{N}'}, D\mu\big)$ and will serve the basis to relate the Holomorphic and Schrödinger representations in the second part of this series. Lastly we introduced the tools of integro-differential calculus, namely the Malliavin derivative and Skorokhod Integral that will represent creation and annihilation operators in the QFT.

 Finally we have described how the choice of Hida test functions to model our quantum states provides us with several tools which will be proven useful in the second part of this paper. Nonetheless, in order to be able to use them in our quantum field model, it is necessary to exhibit in some detail the origin of the geometric elements which is required to incorporate to the corresponding classical model. This is what we introduced in the last section of the present chapter, with particular emphasis in the Kähler structure which will determine the quantization of the theory. We will consider the implications of these objects in the second part of the paper.

 \section{Declarations}

 The authors  declare there is no conflict of interest. 

\section{Acknowledgments}
The authors would like to thank  Prof. Carlos Escudero Liébana for pointing out to very useful bibliography on the connection of Malliavin calculus with QFT. Furthermore, we would like to thank Marc Schneider for stimulating conversations that sparked our curiosity in the phenomenology of quantum completeness in relation with our framework.

The authors acknowledge partial financial support of Grant PID2021-
123251NB-I00 funded by MCIN/AEI/10.13039/501100011033 and by the
European Union, and of Grant E48-23R funded by Government of Aragón.
C.B-M and D.M-C acknowledge financial support by Gobierno de Aragón
through the grants defined in ORDEN IIU/1408/2018 and ORDEN
CUS/581/2020 respectively.

\appendix

\section{Holmorphic Hida test functions}
\renewcommand\thetheorem{\Alph{section}.\arabic{theorem}}
\renewcommand\thelemma{\Alph{section}.\arabic{theorem}}
\renewcommand\thecorollary{\Alph{section}.\arabic{corollary}}
\label{app:Hida}
Most of the references on White noise \cite{kondratievGeneralizedFunctionals1996,obataWhiteNoise1994, hidaWhiteNoise1993, kuoWhiteNoise1996,westerkampRecentResults2003} and Gaussian analysis \cite{huAnalysisGaussian2016} present the theory using Gaussian measures on some particular real DNF $\mathcal{N}^\prime$ or Banach $B$ space. Even tough complex chaos is a straightforward generalization in most cases \cite{hidaBrownianMotion1980}, we prepared in this section a quick dictionary of theorems and lemmas that adapt the classical presentation to our needs.

First let us fix some notation. Consider the system of norms $\lVert\cdot \rVert_{p,q}$ introduced above \eqref{eqn:limitsHida} and their duals $\lVert\cdot \rVert_{-p,-q}$ on the Gel'fand triple $
    \mathcal{H}_{p,q} \subset \mathcal{H}_{Vol} \subset \mathcal{H}_{-p,-q}^*$. Consider also the system of norms $\lvert\cdot \rvert_{p,q}=\lVert\cdot \rVert_{\Gamma\mathcal{H}_{p,q}\otimes \Gamma\mathcal{H}_{p,q}^*}$ of the Fock spaces $\Gamma\mathcal{H}_{p,q}\otimes \Gamma\mathcal{H}_{p,q}^*$. With an slight abuse of notation we identify this space with $\mathcal{I}^{-1}\Big(\Gamma\mathcal{H}_{p,q}\otimes \Gamma\mathcal{H}_{p,q}^*\Big)$ and use the same notation for the norm. We can write  for a function $\Psi(\phi^{\mathbf{x}},\bar \phi^{\mathbf{x}})$ in this space
    \begin{equation*}
        \big\lvert 2^{\frac{q N}{2}} \Psi \big\rvert_{p,0}= \big\lvert  \Psi \big\rvert_{p,q},
    \end{equation*}
where $N$ is the particle number operator \eqref{eq:ChargeAndNumber} of the White noise measure. Recall that the White noise measure is auxiliary and we can describe an equivalent family of norms using the reference measure in \eqref{eq:gelfandtriples}.  
We defined our Hida test functions as holomorphic functions over $\mathcal{N}^\prime_{\mathbb{C}}$. We start by stating the characterization theorem in the holomorphic context:

\begin{theorem} \textbf{(Characterization theorem for holomorphic Hida test Functions \cite{kondratievGeneralizedFunctionals1996})}
    Let $F$ be an entire function on  (holomorphic for every  point in $\mathcal{N}_{\mathbb{C}}$ in the sense of  \autoref{def:Holomorphic}).  Then we have $\overline{F}=S_{\beta}[\Psi]$ where $\Psi\in (\mathcal{N}_\mathbb{C})$ if and only if $\forall\,p,q\geq 0$ there is $C_{p,q}>0$ so that  $F$ fulfils the estimate    \begin{equation*}
        \lvert F(z \rho_{\mathbf{x}})\rvert \leq C_{p,q} \exp\Big(\lvert z \rvert^2\lVert \rho_{\mathbf{x}}\rVert_{-p,-q}^2\Big),\quad \rho_{\mathbf{x}}\in\mathcal{N}_{\mathbb{C}}, z\in \mathbb{C}.
    \end{equation*}
\end{theorem}
\begin{proof} 
Noticing that $S_{\beta}[\Psi]$ is antiholmorphic when $\Psi$ is holomorphic, the proof of Theorem 15 in   \cite{kondratievGeneralizedFunctionals1996} is directly applicable in this context. 
\end{proof}
\begin{corollary}\textbf{(Algebra of Homolorphc Hida test functions )}
     $(\mathcal{N}_\mathbb{C})$ is an algebra under pointwise multiplication. 
\end{corollary}
\begin{proof}
    To show this notice that when $\Psi,\Phi\in (\mathcal{N}_\mathbb{C})$ are holomorphic we have $S_{\beta}[\Psi\Phi](\bar\rho_{\mathbf{x}})=S_{\beta}[\Psi](\bar\rho_{\mathbf{x}})S_{\beta}[\Phi](\bar\rho_{\mathbf{x}})$ thus 

\begin{equation*}
    \lvert S_{\beta}[\Psi\Phi](z\bar\rho_{\mathbf{x}})\rvert \leq   C_{p,q}(\Psi)  C_{p,q}(\Phi)\exp\Big(\lvert z \rvert^2\lVert \rho_{\mathbf{x}}\rVert_{-p,-(q-1)}^2\Big).
\end{equation*}
\end{proof}

For the he general space   $(\mathcal{N}_\mathbb{C})\otimes (\mathcal{N}_\mathbb{C})^*$ this result needs of more involved estimates. We provide an sketch of the proof adapting the study found in \cite{obataWhiteNoise1994}. First notice that, in a similar way to  \cite{obataWhiteNoise1994} Lemma 3.5.1:
\begin{multline}
 e^{\overline{\rho_z\phi^z}+\rho_z\phi^z-\bar\rho_u\delta^{uv}\rho_v}
 e^{\overline{\chi_z\phi^z}+\chi_z\phi^z-\bar\chi_u\delta^{uv}\chi_v}
=\\
e^{\overline{(\rho_z+\chi_z)\phi^z}+(\rho_z+\chi_z)\phi^z-\overline{(\rho_u+\chi_u)}\delta^{uv}(\rho_v+\chi_v)}e^{\overline{\rho_u}\delta^{uv}\chi_v+\overline{\chi_u}\delta^{uv}\rho_v}
\end{multline}
and \eqref{eq:ComplexWienerIto} imply, 
 \begin{multline*}
    \mathcal{W}_\delta(\phi^n\bar{\phi}^{\bar{n}})^{\vec{\mathbf{x}}_n\vec{\mathbf{y}}_{\bar{n}}}\mathcal{W}_\delta(\phi^m\bar{\phi}^{\bar{m}})^{\vec{\mathbf{x}}_m\vec{\mathbf{y}}_{\bar{m}}}=\\
    \sum_{k=0}^{\min({n},\bar{m})}\sum_{\bar k=0}^{\min(\bar{n},{m})} \binom{n}{k}\binom{\bar{m}}{k} k! \binom{\bar{n}}{\bar k}\binom{{m}}{\bar k} \bar k! \mathcal{W}_\delta(\phi^{n+m-k-\bar{k}}\bar{\phi}^{\bar{m}+\bar{n}-k-\bar{k}})^{(\vec{\mathbf{x}}_\cdot;(\vec{\mathbf{y}}_{\cdot}} \big[\delta^{\bar k}\big]^{ \vec{\mathbf{y}}_{\bar{k}})}
    \big[\delta^{k}\big]^{ \vec{\mathbf{x}}_{{k}})}
\end{multline*}
where the symmetrization is performed over superindices with the same label. 
We want to adapt the arguments leading to Theorem 3.5.6 in \cite{obataWhiteNoise1994} in a similar way as Theorem 64 in  \cite{westerkampRecentResults2003}. 
\begin{theorem}\textbf{(Algebra of complex Hida test functions)} 
    From continuity we get 
that there is a $K_p\geq 0$ such that   ${\overline{\rho}_{{x}}\delta^{xy}\rho_{{x}}}= \lVert \rho_{\mathbf{x}} \rVert_{0,0}^2\leq K_p \lVert \rho_{\mathbf{x}} \rVert_{p,0}^2 $. Choose $\alpha, \beta$ such that, for a fixed $q,p\geq 0$ we get $2^{-(\alpha+\beta)/2}K_p^{2} -2^{q-\alpha}-2^{q-\beta}< 1$ then we estimate 

\begin{equation}
    \lvert \Psi\Phi\rvert_{p,q} \leq \frac{1-2^{-(\alpha+\beta)/2 }K_p^{-2}}{(1-2^{-(\alpha+\beta)/2}K_p^{2} -2^{q-\alpha}-2^{q-\beta})^2} \lvert \Psi\lvert_{p,\alpha} \lvert\Phi\rvert_{p,\beta}
\end{equation}
And $(\mathcal{N}_\mathbb{C})\otimes (\mathcal{N}_\mathbb{C})^*= \bigcap\limits_{p,q} \mathcal{I}^{-1}\Big( \Gamma\mathcal{H}_{p,q}\otimes \Gamma\mathcal{H}_{p,q}^* \Big)$ is an algebra under pointwise multiplication. \end{theorem}

\begin{lemma}
\label{lemm:Banach} (Equivalence to the  Wiener space)
    Consider the Gaussian measure defined in  \eqref{eq:complexCharacteristicFunctional}. Because of continuity with respect to the topology of $\mathcal{N}_{\mathbb{C}}$, generated by the family of norms $\{ \lVert \cdot \rVert_{p}\}_{p=0}^\infty$, we get that  for some $p,q\geq 0$ we must have  $\lvert \bar\rho_{x}\Delta^{xy}\rho_y\rvert \leq \lVert \rho_{\mathbf{x}} \rVert_{p,q}$. Because of the definition of FN space \autoref{def:FN} the norms are taken such that $\forall\, q> p$ the inclusion $\mathcal{H}_q\hookrightarrow \mathcal{H}_p$ is Hilbert-Schmidth. As a result the measure is supported in ${\mathcal{H}}_q^\prime=\mathcal{H}_{-q}\subset \mathcal{N}_{\mathbb{C}}'$. This is 
    \begin{equation*}
        \mu_c(\mathcal{H}_{-q})= \mu_c(\mathcal{N}_{\mathbb{C}}')= 1.
    \end{equation*}
\end{lemma}

\begin{proof}
   This is shown in Theorem 3.1 of \cite{hidaBrownianMotion1980} . 
\end{proof}

We will use \autoref{lemm:Banach} to apply results derived in \cite{huAnalysisGaussian2016} for Gaussian measures on Banach spaces, this is, using the Gel'fand triple $\mathcal{H}_{q}\subset\mathcal{H}_{\Delta}\subset \mathcal{H}_{-q}$. In particular we get 

\begin{theorem}{(\textbf{Hypercontractivity})} Let $1 \leq p \leq q < \infty$ then we get $2^{-\alpha\frac{{N}}{2}}$ is a bounded linear mapping such that 

\begin{equation}
   \lVert  2^{-\alpha\frac{{N}}{2}} \Psi \rVert_{L^q} \leq \lVert \Psi \rVert_{L^p}
\end{equation}
if and only if $2^{-\alpha} (p-1)\leq (q-1)$.    
\end{theorem}
\begin{proof}
    Theorem 7.2 of \cite{huAnalysisGaussian2016}
\end{proof}

\begin{corollary}\textbf{(Integrability properties)} As a result of hypercontractivity we get 
$(\mathcal{N}_\mathbb{C})\subsetneq \bigcap\limits_{p\geq1, m\geq 0}\mathbb{D}^{p,m}_{\mu_c}= \mathbb{D}_{\mu_c}$. 
\end{corollary}
\begin{proof} Let $p\geq 1$ choose $\lambda\geq 0$ such that $2^{-\lambda} (p-1)\leq 1$ then 
\begin{align*}
    \lVert (1+N)^{\frac{m}2} \Psi \rVert_{L^p}&= \lVert 2^{-\lambda \frac{N}2}2^{\lambda \frac{N}2}(1+N)^{\frac{m}2} \Psi \rVert_{L^p}\\
    &\leq \lVert 2^{\lambda \frac{N}2}(1+N)^{\frac{m}2} \Psi \rVert_{L^2} \\
    & \leq \lVert 2^{\lambda \frac{N}2}(1+2^N)^{\frac{m}2} \Psi \rVert_{L^2} \\
    & \leq 2^{\frac{m}2}\lVert 2^{(\lambda+m)\frac{N}2} \Psi \rVert_{L^2}= 2^{\frac{m}2} \lvert \Psi \rvert_{0,\frac{\lambda+m}2} 
\end{align*}
Where in the first inequality we used hypercontractivity for $p\geq 2$ and for $1\leq p \leq 2$ we use  Hölder inequality and the fact that $\lVert 2^{-\lambda \frac{N}2}\rVert\leq 1$ in the operator norm. This proof is valid also using the reference measure $\mu_c$ to define the topology of $(\mathcal{N}_\mathbb{C})$.  This proves that, defining the auxiliary space $\mathcal{G}_{\mu_c}= \bigcap\limits_{q} \mathcal{I}_{\mu_c}^{-1}\Big( \Gamma\mathcal{H}_{0,q}\Big)$,  
$$(\mathcal{N}_\mathbb{C})\subsetneq \mathcal{G}_{\mu_c} \subsetneq\bigcap\limits_{p\geq1, m\geq 0}\mathbb{D}^{p,m}_{\mu_c}= \mathbb{D}_{\mu_c}$$
Further properties of $\mathcal{G}$ can be found in \cite{potthoffDualPair1995}. 
\end{proof}


\end{document}